\documentclass[aps,superscriptaddress,prd,preprintnumbers,twocolumn,floatfix]{revtex4}
\pdfoutput=1
\usepackage{bm}
\usepackage{latexsym}
\usepackage{dcolumn}
\usepackage{amsmath,amsfonts,amssymb}
\usepackage{graphicx}
\usepackage{psfrag}
\usepackage{amsthm}
\usepackage{color}
\usepackage[T1]{fontenc}
\usepackage[utf8]{inputenc}
\usepackage{lmodern}
\usepackage{natbib}

\def\tr {\operatorname{tr}}
\def\fat{\mathbf}
\def\bra#1{\langle #1 |}
\def\ket#1{| #1 \rangle}

\def\no {\noindent}

\def\k  {\fat k}

\def\p  {\fat p}
\def\x  {\fat x}
\def\H  {\fat H}
\def\aa  {\fat a}

\def\hannover {Institut f\"ur Theoretische Physik, Leibniz Universit\"at Hannover, Hannover, Niedersachsen, Germany}

\newcommand{\mc}{\mathcal}
\newcommand{\mr}{\mathrm}
\newcommand{\mb}{\mathbb}

\newcounter{counter}
\newtheorem{cor}[counter]{Corollary}
\newtheorem{theorem}[counter]{Theorem}

\newtheorem*{assumptions}{Assumptions}
\newtheorem{proposition}{Proposition}
\newtheorem{definition}{Definition}

\begin{document}

\title{Optimal uncertainty relations in a modified Heisenberg algebra}

\author{Kais Abdelkhalek}\email[\ ]{kais.abdelkhalek@itp.uni-hannover.de}
\address{\hannover}

\author{Wissam Chemissany}\email[\ ]{wissam.chemissany@itp.uni-hannover.de, wissamch@mit.edu}
\address{\hannover}
\address{Research Laboratory of Electronics, Massachusetts Institute of Technology, Cambridge, Massachusetts, USA}

\author{Leander Fiedler} \email[\ ]{leander.fiedler@itp.uni-hannover.de}
\address{\hannover}

\author{Gianpiero Mangano}\email[\ ]{mangano@na.infn.it}
\address{INFN, Sezione di Napoli, Complesso Univ. Monte S. Angelo, Napoli, Italy}

\author{René Schwonnek} \email[\ ]{rene.schwonnek@itp.uni-hannover.de}
\address{\hannover}

\begin{abstract}
Various theories that aim at unifying gravity with quantum mechanics suggest modifications of the Heisenberg algebra for position and momentum. From the perspective of quantum mechanics, such modifications lead to new uncertainty relations which are thought (but \emph{not}  proven) to imply the existence of a minimal observable length. Here we prove this statement in a framework of sufficient physical and structural assumptions. Moreover, we present a general method that allows to formulate optimal and state-independent variance-based uncertainty relations. In addition, instead of variances, we make use of entropies as a measure of uncertainty and provide uncertainty relations in terms of min- and Shannon entropies. We compute the corresponding entropic minimal lengths and find that the minimal length in terms of min-entropy is exactly one bit.
\end{abstract}

\maketitle


\section{Introduction}
A considerable amount of efforts has been devoted to reconcile gravity with quantum mechanics, but the conventional field theoretic avenues for quantizing general relativity have suffered issues in renormalisability. Several theories such as string theory have suggested that the sought-after quantum gravity has to be effectively cut off in the ultraviolet, leading to the notion of \emph{minimal length} \cite{Amati:1988tn,garay1995quantum,kempf1995hilbert} (see  also \cite{hossenfelder2013minimal} and references therein). In other words, the gravitational effects become significantly important upon probing physics at an energy scale as large as the Planck scale.
Such a nontrivial premise of the minimal position uncertainty has been corroborated by string theoretic arguments \cite{Amati:1988tn,Konishi:1989wk}, leading to the so-called generalized uncertainty principle.

There had been a consensus within the high energy physics community that such a minimal length has a quantum mechanical origin which should effectively be formulated in the form of a non-zero minimal uncertainty for a position measurement. In its simplest version, this can be obtained by explicitly constructing position and momentum operators $\x$ and $\p$ that satisfy a deformed Heisenberg algebra
\begin{align}
[\x,\p]=i \hbar f(\p),
\label{modalg0}
\end{align}
where the precise form of the modification $f(\p)$ depends on which theory and approach is used \cite{Amati:1988tn,Konishi:1989wk,Kempf:1993bq,connes}, see also e.g. \cite{Moyal:1949sk, Groenewold:1946kp, Doplicher:1994tu} for deformations in configuration space. Recently, a deformation of (\ref{modalg0}) where the r.h.s. is assumed to be a stochastic gaussian variable has been also considered \cite{Mangano:2015pha}. Large parts of related literature focuses on the case where $f(\p)=1+\beta \p^2$ and uncertainty is measured in terms of variances. This specific modification corresponds to the first term in a Taylor expansion in $\p^2$ and reflects the expected simplest deviation from the standard case. In this work we allow modifications that satisfy a number of physically well-motivated assumptions (see Section \ref{heisenalg}) and are otherwise completely general.

Despite the substantial understanding that had been gained in previous approaches, there still exist conceptual shortcomings in the study of the origin of minimal lengths. Firstly, the term ``minimal length'' should not be interpreted as an actual geometric length. Instead, minimal length refers to the perception that, in a modified algebra, the probability distributions obtained in a position measurement cannot become arbitrarily sharp. Hence, minimal length should rather be dubbed and always be understood as ``minimal position uncertainty''.

\begin{figure}[t]
	\centering \def\svgwidth{0.45\textwidth}
	\input{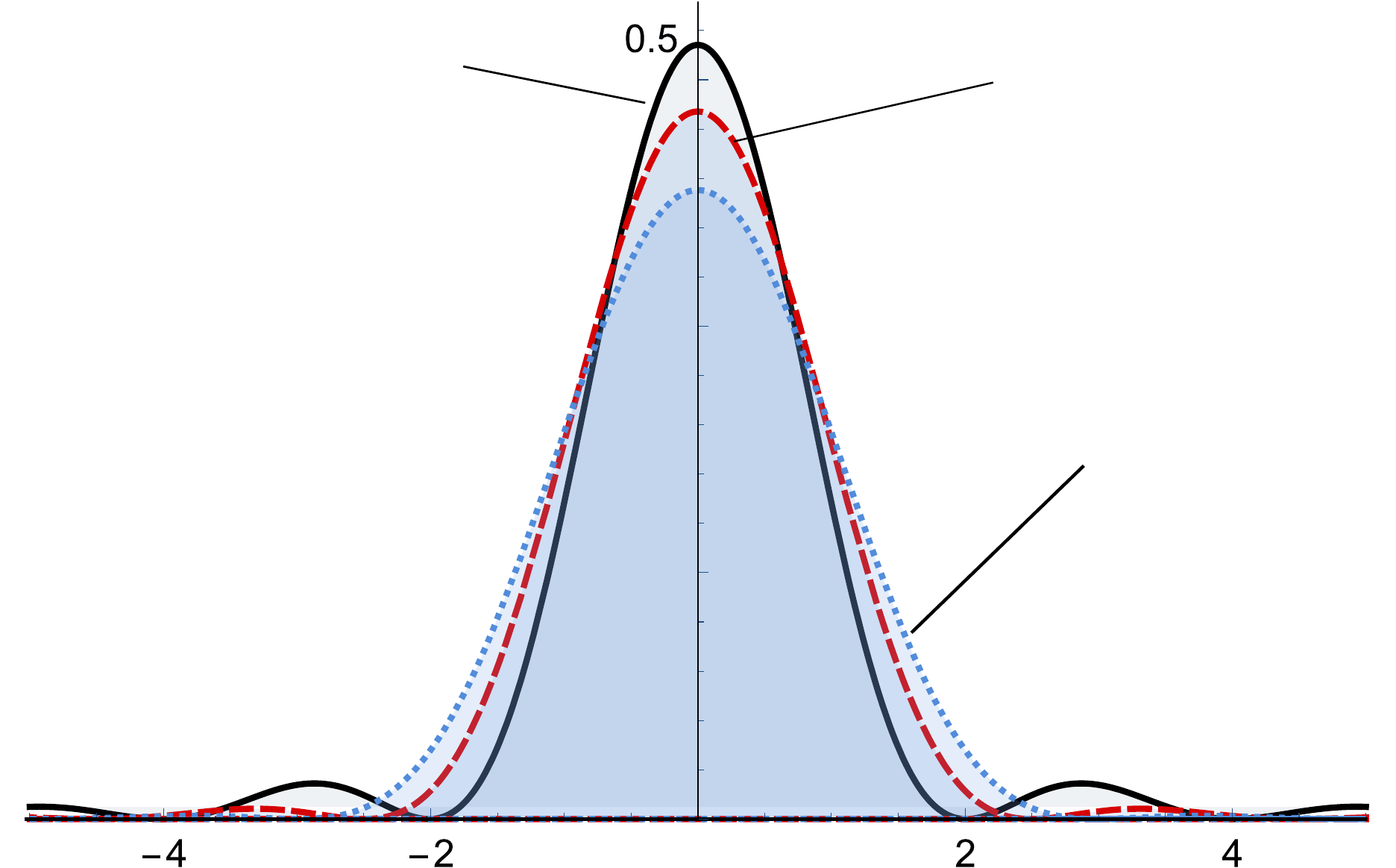}
	\caption{
		Probability distributions of a position measurement on quantum states suspected to a UV cut-off (see  \eqref{optstates}). Which one is the ``sharpest''?
		All three distributions attain a minimal length in their own sense:
		the blue, dotted one attains minimal length in terms of variances as typically considered in the literature. The black, solid distribution has the highest peak and the red, dashed distribution has the smallest entropy. The latter two distributions have infinite variance, although they appear to be localised to some extent.}
	\label{posprob}
\end{figure}

Secondly, if minimal length is to be understood as an immediate consequence of modifying the Heisenberg algebra, it is important to show that {\it all} possible pairs of operators $\x$ and $\p$ that satisfy this algebra lead to a non-zero minimal position uncertainty. If such a statement is not correct in all its generality, what assumptions are needed besides modifying the algebra to prove minimal length? In the standard case this question is answered by the Stone-von Neumann theorem, in the present context the situation is not at all clear. Most work related to the study of minimal length focused only on showing the mere existence of such operators neglecting such uniqueness considerations\footnote{While there have been studies \cite{kempf1995hilbert,Kempf:1997pb} that tackle the uniqueness of self-adjoint extensions of the position operator given a fixed representation of position and momentum, we discuss the uniqueness of representations of a modified momentum operator as a function of the unmodified momentum operator.}.

Thirdly, no general method to compute optimal and state-independent uncertainty relations for a given modification $f(\p)$ has been developed so far. While such relations directly yield minimal length, they are of scientific interest on their own, since they express the influence of the modification on all states. For example, experimental proposals like \cite{Pikovski:2011zk} aim at observing a modified uncertainty relation in an uncertainty regime where minimal length cannot be attained.

As a last point, characterizing minimal length in terms of variances is at least controversial: on one hand variances characterise well the uncertainty for most unimodal distributions, especially if they are Gaussian. On the other hand, variances of multimodal distributions are known to show strange and unwanted behaviour if interpreted as a measure of uncertainty (as is done for minimal lengths). There is no reason why a position distribution should in general be unimodal, especially since non-zero minimal length immediately implies that Gaussian states are not part of the considered Hilbert space. Hence, in most cases variances are not a good candidate to capture the notion of minimal length as can also be seen from Fig.~\ref{posprob}. In Section~\ref{sec:entrops} we discuss this in more detail.

Which measure to use instead is far from unique and depends on the operational task that shall be accomplished. Entropies as a measure of uncertainty have been proven tremendously useful in various fields, such as quantum information theory (see e.g. \cite{Wilde}), quantum thermodynamics (see \cite{Goold:2015} for a survey), or quantum gravity (for recent work and references see \cite{carroll2016entropy}). For example, the uncertainty principle has been made operationally precise in the form of entropic uncertainty relations which for instance are an essential part of security proofs of quantum cryptographic protocols (e.g. \cite{Furrer}) and are still focus of much investigation \cite{dammeier2015uncertainty, Busch:2013vba,abdelkhalek2015optimality}, see also the reviews \cite{coles2015entropic,WehnerWinter}.
It thus seems beneficial to formulate entropic uncertainty relations in the context of modified Heisenberg algebras thereby introducing the concept of {\it entropic minimal length}. We investigate its implications compared to those obtained by its variance-based counterpart, see \cite{Pye:2015tta}.

In this work we develop the underlying quantum mechanical setting in which one can study the {\it direct} consequences of modifying the Heisenberg algebra in view of the existence of non-zero minimal lengths. Then, temporarily complying with the consensus to formulate minimal length in terms of variances, we provide a general framework from which optimal and state-independent uncertainty relations and minimal lengths can be calculated efficiently.
Here, the term ``optimal'' refers to Pareto optimality, which originates in the theory of optimisation \cite{Pareto}.
We will also argue why a typical approach that invokes equality in the Robertson-Kennard relation \eqref{rob} does not provide an optimal uncertainty relation. Lastly, we compute and discuss implications and advantages of an entropic formulation of minimal length. In particular, we introduce minimal length in terms of the Shannon entropy (or differential entropy in the continuous setting) and show that both the minimal length value and the corresponding minimizing states are not equivalent to those obtained for the ``standard'' minimal length in terms of variances. We also discuss how a further feature appears when using entropies, namely a maximal entropy in momentum space. Finally, we compute minimal length in terms of the min-entropy, which quantifies the maximum probability of correctly predicting the outcome of a position measurement. Intriguingly, we find an intimate connection between variance-based and min-entropy based minimal length: for scenarios with normalised variance-based minimal length, the mininal length in terms of min-entropy is also normalised, meaning that the best possible localisation of space is exactly one bit.

\par\noindent
\section{Representation of the modified Heisenberg Algebra}
\label{heisenalg}
Let us consider the position operator $\x$, i.e. the multiplication operator
on the Hilbert space $\mathcal{H}:=\mathcal{L}^2(\mathbb{R})$. In this section we characterise properties of momentum operators $\p$ that satisfy the modified Heisenberg algebra \eqref{modalg0}. After briefly summarising and unifying previous constructions that showed the {\it existence} of operators $\p$ that lead to minimal lengths effects, we present our main result of this section that proves the {\it uniqueness} of such constructions.

More concretely, we aim at characterising a linear, self-adjoint operator $\p$ with dense domain in a closed subspace $\mathcal{P}$ of a Hilbert space $\mathcal{H}$ with spectrum coinciding with $\mathbb{R}$ that satisfies the modified Heisenberg algebra\footnote{Note that in the following $\fat x$ is always the usual position operator on $\mc H$ in contrast to the approach taken in \cite{kempf1997minimal,spindel}. However, we can obtain the form of the position operators in these references, if we restrict $\fat x$ to $\mc P$ and then choose suitable self-adjoint extensions.} (we set $\hbar=1$ in the following)
\begin{align}
	[\x,\p]=if(\p) \ ,
	\label{modalg}
\end{align}
where $f:\mathbb{R}\rightarrow \mathbb{R}$ satisfies
\begin{itemize}
	\item[(i)] $f(0)=1 \ ,$
	\item[(ii)] $f(p)=f(-p)$ for all $p\in\mathbb{R}\ , $
	\item[(iii)] $f(p)$ is convex on $\mathbb{R}^+$, i.e. $\forall p,p'\geq0$: \\
	$f(\lambda p + (1-\lambda) p')\leq \lambda f(p) + (1-\lambda) f(p')$.
\end{itemize}
Assumption (i) ensures that for small momentum we retrieve the original unmodified Heisenberg algebra, while assumption (ii) translates to the statement that momentum should not have a preferred direction. The last assumption (iii) is a generalisation of modifications that were considered previously in the literature \cite{hossenfelder2013minimal,kempf1995hilbert} and implies that higher momentum leads to stronger effects of the modification.

Since $f(\p)$ is adimensional, it depends on momentum via the product $\sqrt{\beta} \, \p$ , with $\beta$ a constant with dimension of inverse squared momentum (or inverse squared mass in natural units), which sets the scale where deviations with respect to the standard picture are important\footnote{We use here the same notation of \cite{kempf1995hilbert} to simplify the comparison of our results with those described there.}. In natural units $\beta$ is naturally expected to be of the order of $m_{Pl}^{-2}$, with $m_{Pl}$ the Planck mass, but we consider this scale as a free parameter.

Previous works aimed at explicitly constructing operators $\p$ that satisfy the algebra \eqref{modalg} and lead to a non-trivial minimal length. Before discussing subtleties arising in these approaches, we briefly review these constructions which may be unified as follows:
consider the {\it unmodified} momentum operator $\k$ on $\mathcal{H}$ such that $\x$ and $\k$ satisfy the standard commutation relation
\begin{equation}
	[\x,\k]=i\mathbb{I}\ ,\label{unmod}
\end{equation}
i.e. the momentum operator is given by
\begin{align}
\k=\mathcal{F}^{\dagger} \x \mathcal{F}\ ,
\end{align}
with $\mathcal{F}$ the Fourier transform acting on states $\phi\in\mathcal{H}$ via
\begin{align}\label{fourier}
(\mathcal{F} \phi)(k)=\frac{1}{\sqrt{2\pi}}\int \mathrm{e}^{ikx} \phi(x) dx\ .
\end{align}
We can then deform the spectrum of $\k$ until we find a linear operator $\p=p(\k)$ that satisfies \eqref{modalg} (see Fig. \ref{kspace}).
\begin{figure}[t]
	\centering
	\def\svgwidth{0.38\textwidth} \input{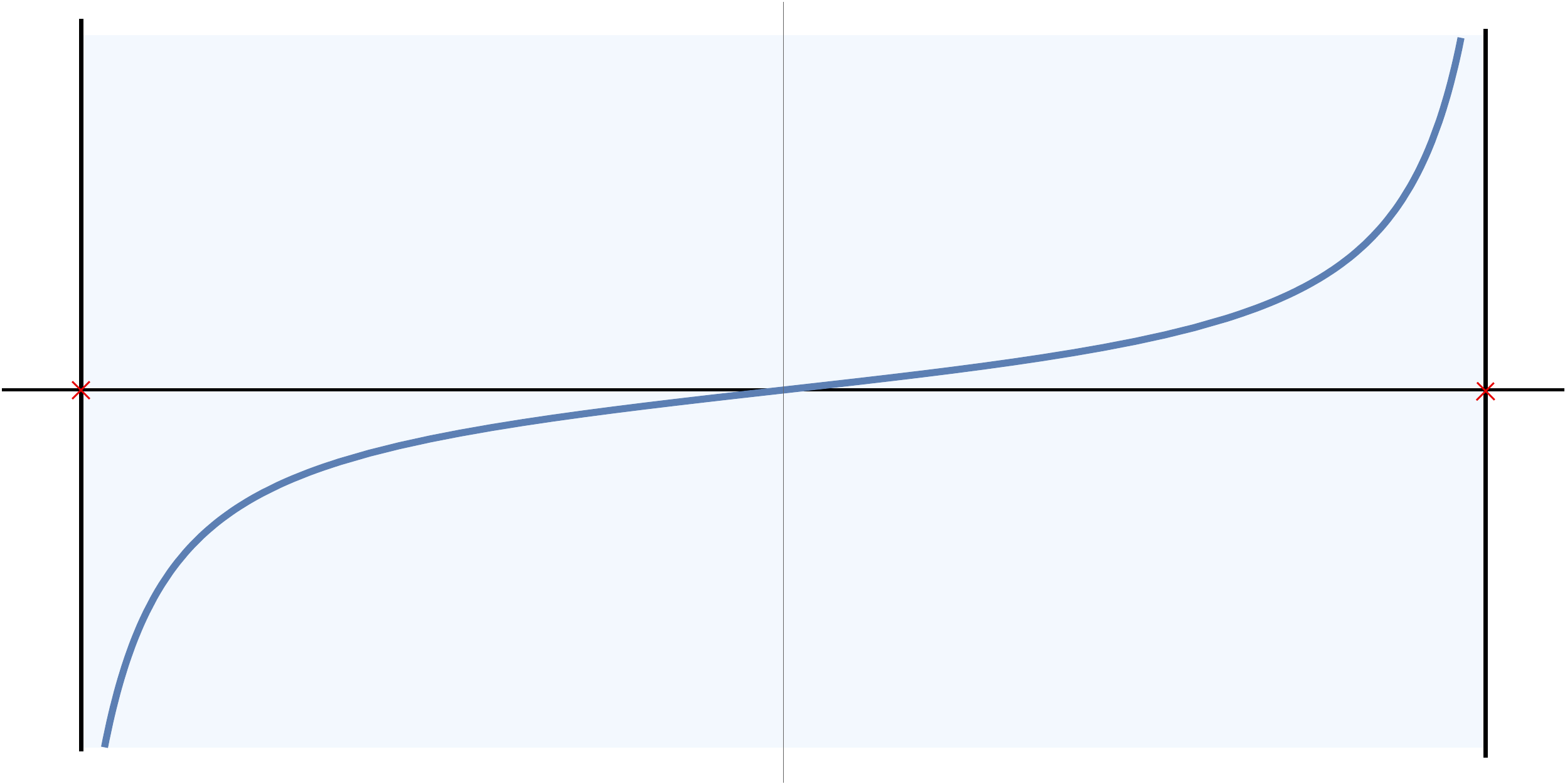}
	\caption{Theorem \ref{thm:mainthm} gives sufficient conditions for representing the operator $\p$ as a function $p(\k)$, where $\k$ denotes the unmodified momentum operator (here shown for the modification $f(p)=1+ \beta p^2$). We show that the support of $\k$ must be restricted to an interval $[-k_{max},k_{max}]$, i.e. modifying the algebra directly leads to UV cut-off and, hence, minimal length.}  \label{kspace}
\end{figure}
More concretely, by functional calculus we can evaluate the commutator
\begin{equation}\label{xfp}
	[\x,\p]=\left[i\frac{d}{dk},p(k)\right]=i\frac{d}{d k} p(k)\ ,
\end{equation}
to find that \eqref{modalg} translates to the differential equation
\begin{align}\label{fkeq}
	\frac{d}{d k} p(k)=f(p(k))\ .
\end{align}
By the implicit function theorem we obtain the solution
\begin{align}
\label{kofp}
k(p)=\int_{p_{0}}^{p} dp' \frac{1}{f(p')} \ ,
\end{align}
where we set $p_0=0$, such that the momentum operators $\p$ and $\k$ yield the same physics in the small-momentum regime. For all cases where this integral is finite in the limit $p \rightarrow \infty$, this implies the existence of a momentum cut-off,
\begin{equation}
	k_{\max}:=\int_{0}^{\infty} dp' \frac{1}{f(p')}\ .\label{eq:kmax}
\end{equation}
This argument, commonly found in related literature \cite{kempf1995hilbert,Brout}, shows that for states with support in the interval $[-k_{\max},k_{\max}]$ there exist operators $\fat x$ and $\fat p$ satisfying the modified commutation relation~\eqref{modalg}.
Moreover, the operator $\fat p$ is just defined on a proper subspace $\mc P=\mc L^2([-k_{\max},k_{\max}])$ of the Hilbert space $\mc H$ (Fig.~\ref{fig:sets}).
In particular, states with vanishing position uncertainty which have, by \eqref{fourier}, a broad momentum distribution are no longer contained in $\mc P$, hence implying the existence of a non-trivial minimal length.
\begin{figure}[h]
	\centering
	\def\svgwidth{0.3\textwidth} \input{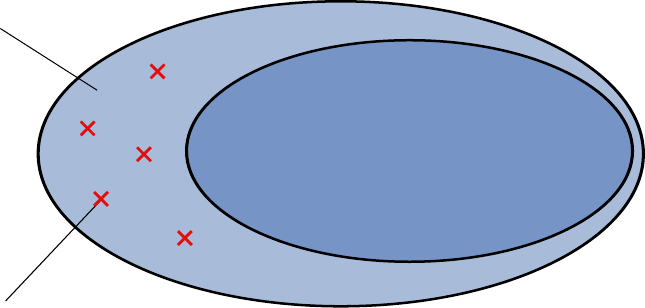}
	\caption{For an unmodified algebra, minimal uncertainty is obtained for states $\phi$ from a state space $\mathcal{H}$. In the case of the usual Heisenberg algebra the corresponding minimal length becomes trivial, that is, all states in $\mathcal{H}$ are physical. If a modification of the algebra directly leads to a restricted state space $\mathcal{P}$ that does not contain such states, one obtains non-trivial minimal length as a direct consequence of modifying the algebra.
	}
	\label{fig:sets}
\end{figure}
In this sense the existence of a momentum cut-off, sometimes also referred to as a UV cut-off, directly implies the existence of a non-trivial minimal length. Conversely, if $\mathcal{P}=\mathcal{H}$, i.e. there is no momentum cut-off, there are states with vanishing position uncertainty as is the case for the unmodified Heisenberg algebra.

However, in order to interpret the momentum cut-off and the corresponding minimal length as a direct consequence of modifying the algebra, it is not sufficient to just show the existence of operators that allow for non-trivial minimal lengths as above. Instead, one needs to show that {\it all} operators $\x$ and $\p$ satisfying \eqref{modalg} lead to this effect. In this section we show that this is indeed the case if, given $\x$ is the standard position operator, we additionally require that the spectral projections of the canonical momentum operator $\k=\mathcal{F}^{\dagger} \x \mathcal{F} $ and the modified momentum operator $\p$ are close to each other in the regime of small momenta. This assumption immediately implies that, in this regime, the probability distributions induced by $\k$ and $\p$ are almost the same, which agrees with the intuition that {\it observable effects of the modified algebra only occur for high momenta}. More precisely, we require that there exists $\epsilon_0> 0$ such that for all $\epsilon\in(0,\epsilon_0)$ there is an $\epsilon'>0$ and $\delta>0$ with $\delta \sim \mc{O}(\epsilon^3)$ such that 
\begin{align}\label{eq:assump1}
\|E_{\fat k}([-\epsilon,\epsilon])-E_{\fat p}([-\epsilon',\epsilon'])\|<\delta \ ,
\end{align}
where $E_{\fat k}$ and $E_{\fat p}$ denote the spectral projections of $\k$ and $\p$, respectively. With this assumption we show the following theorem (see Appendix for the proof):

\begin{theorem}
	\label{thm:mainthm}
Let $\x$ be the position operator and $\k=\mathcal{F}^{\dagger} \x \mathcal{F}$ be the unmodified momentum operator as defined above. Denote by $\p$ a modified momentum operator on a Hilbert space $\mathcal{P}$, i.e. $\x$ and $\p$ satisfy the modified Heisenberg algebra \eqref{modalg} for all states in $\mathcal{P}$. If additionally \eqref{eq:assump1} is satisfied, then 
\begin{itemize}
	\item there exists a \emph{momentum cut-off}, i.e. there is an interval
	\begin{equation}
	I=[-k_{\max},k_{\max}]\subseteq\mathbb{R}\ ,
	\end{equation}
	with $k_{\max}$ as in \eqref{eq:kmax} such that $\mathcal{P}=\mathcal{L}^2(I)$,
	\item there is a function $p:I\rightarrow\mathbb{R}$ such that for all states in $\mathcal{P}$
	\begin{equation}
	\p=p(\k) \ .
	\end{equation}
	Hence, $\p$ is indeed a function of $\k$, which legitimates the standard construction after eq. \eqref{fourier}.
\end{itemize}
\end{theorem}

In other words, the scope of Theorem \ref{thm:mainthm} can be summarised as follows: assume an experimenter who, on a length scale that is above Planck length, can agree on a clear notion of what the position $\x$ and the unmodified momentum $\k$, i.e. a particular representation of the Heisenberg algebra, should be. If he extrapolates his notion of $\x$ down to lower scales, and assumes that a modified algebra has a  consistent limit to what he observed on higher scales, he obtains by Theorem~\ref{thm:mainthm} a {\it unique} notion of what $\p$ is in this situation.
Theorem \ref{thm:mainthm} can therefore be seen as the reason why the aforementioned construction is indeed meaningful: the construction describes {\it all} possible modified momentum operators $\p$. Importantly, it proves the existence of a UV cut-off and a corresponding non-trivial minimal length as a direct consequence of modifying the underlying algebra.

Note that Theorem \ref{thm:mainthm} builds on the natural assumption that measurement probabilities should be similar when measuring the modified momentum operator $\p$ or the unmodified momentum operator $\k$ in the regime of small momentum. This assumption is essential since it provides a means to characterise the action of $\x$ on states in $\mathcal{P}$: as the position operator $\x$ induces shifts in $\k$-momentum space, knowing that $\p$ and $\k$ are not too different in the small-momentum regime implies that $\x$ also induces shifts in $\p$-momentum space (up to some arbitrarily small error). By how much $\x$ is shifting a state in $\k$- or $\p$-space is governed by the respective commutation relation. Hence, while $\x$ induces constant shifts in $\k$-space, the strength of shifting in $\p$-space is monotonically increasing with higher momentum as is the modification $f$ (Assumption (iii)). This leads to normalizability constraints: for high enough momentum the shift becomes too large for the corresponding states to be normalizable. The cut-off parameter $k_{\max}$ is exactly the momentum value for which the states cannot be normalised anymore.

The self-adjoint position operator $\x$ has a domain dense in $\mathcal{H}$. Theorem \ref{thm:mainthm} shows however that, if the modification $f(\p)$ is such that $k_{\max}$ is finite, the relevant Hilbert space $\mathcal{P}$ is strictly smaller than $\mathcal{H}$. So, how can the measurements of position be described if acted on states in $\mathcal{P}$? This is in particular interesting since the position operator $\x$ restricted to $\mathcal{P}$ is {\it not} self-adjoint anymore. Nevertheless, a position measurement still has a well-understood description, known as a POVM (positive operator valued measure). POVMs describe the most general form of a quantum measurement (see e.g. \cite{Wilde}). An explicit construction of the position operator as a POVM on a restricted state space can be found in \cite{Werner1990}.

\section{Optimal uncertainty relations in terms of variances}
The concept of minimal length expresses the fact that for all states position measurements will in general not produce arbitrarily sharp outcome distributions. This is typically quantified by computing the minimal variance of this distribution
\begin{equation}
l^2_{\min}= \min_{\psi\in\mathcal{P}} \Delta\x \ ,
\label{minlengthdef}
\end{equation}
where
\begin{align}
\Delta\x  = \langle \psi | \x^2 |\psi\rangle - \langle \psi | \x |\psi\rangle^2.\end{align}
As such, minimal length is intimately related to the concept of uncertainty relations, which place constraints on how sharp the distribution for some observable $A$ can be, given the sharpness of the distribution of another, say $B$. The standard example of such an uncertainty relation is the one due to Robertson and Kennard \cite{Robertson,Kennard}, i.e.
\begin{align}
\label{rob}
\Delta A \Delta B \geq \frac{1}{4}|\bra{\psi} [A,B] \ket{\psi} |^2 \ .
\end{align}
A naive approach to compute the minimal length is to impose equality in \eqref{rob} and then search for minimising states within the corresponding subset of states. In this section we will remark that this approach will fail in most cases due to the state dependence of the lower bound in \eqref{rob}. Instead, we provide a general framework to obtain optimal and state-independent uncertainty relations for a modified Heisenberg algebra in the first part of this section. As a side product this  allows to directly compute the corresponding minimal length. Then, in the second part of the section, we exemplarily apply this framework to the modification $f(p)=1+\beta p^2$, since here all differential equations can be solved analytically. We obtain the same uncertainty relation as in \cite{kempf1995hilbert}, but now with a proof of its optimality. We also apply our framework to other modifications, i.e $f(p)=\cosh(\sqrt{\beta}p)$ and $f(p)=1+ \beta p^2+\beta^2 p^4/4$, employing numerical tools (see also \cite{Bojowald:2011jd} for a treatment of higher order modifications $f(\p)$).

\subsection{A general method for finding uncertainty relations and minimal length}
Uncertainty relations allow to lower bound the uncertainty of one measurement, given the uncertainty of another measurement. Having this in mind, a good way to generally think about uncertainty is in terms of diagrams, as shown in Fig. \ref{dxdp} \citep{dammeier2015uncertainty,kempf1995hilbert, abdelkhalek2015optimality}. Here, the blue shaded region indicates the set $\mathcal{U}$ of all tuples $(\Delta \p,\Delta \x)$ that can be obtained by measuring both $\p$ and $\x$ on the same state $\psi$, where $\psi$ is taken from $\mathcal{P}$. In the following we will refer to $\mathcal{U}$ as the \textit{uncertainty region}.

\begin{figure}[htb]
\centering \def\svgwidth{0.45\textwidth}
\input{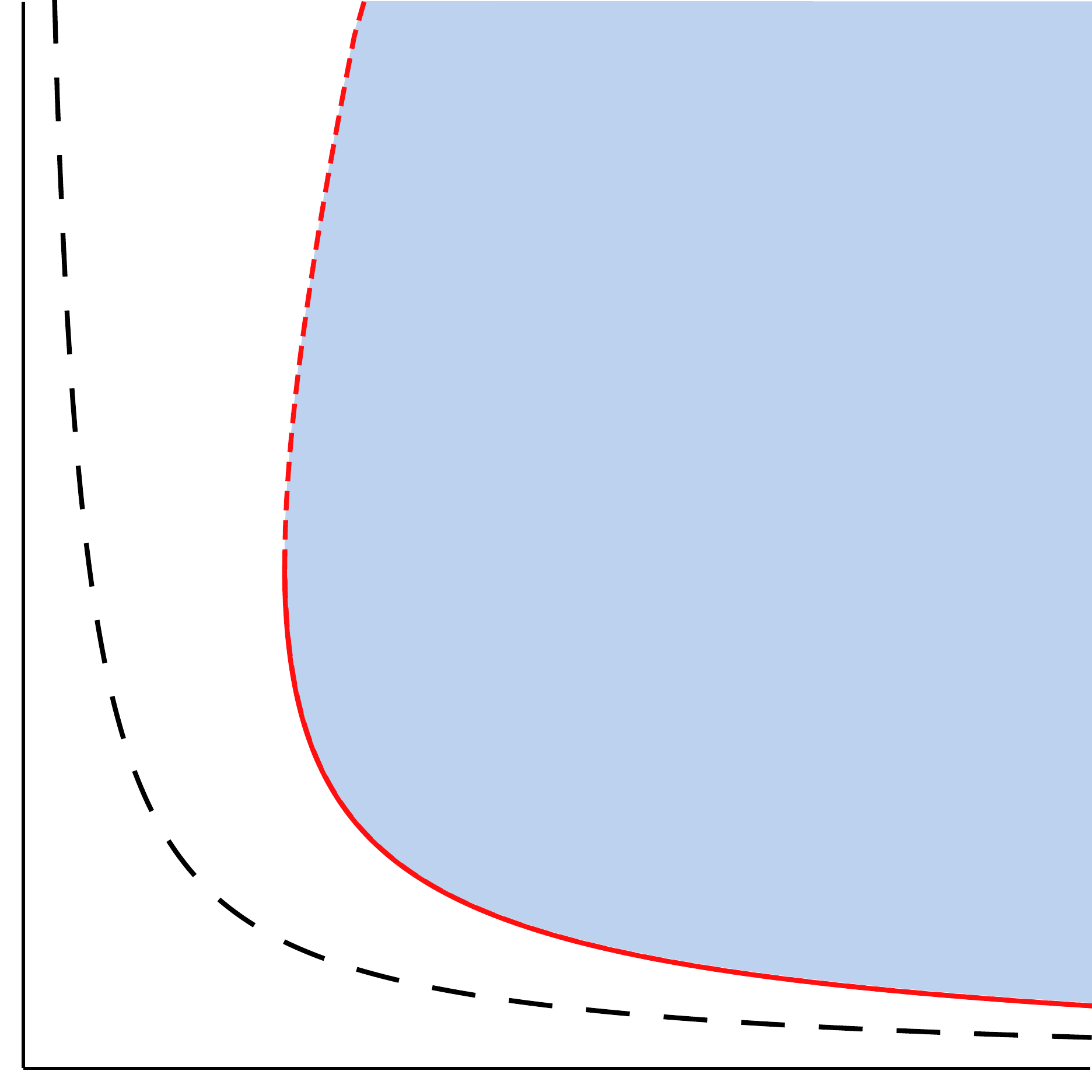}
\caption{The boundary of $\mathcal{U}$ for the modified (with $f(p)=1+\beta p^2$) and standard (long dashed black curve) Heisenberg algebra. We show the point corresponding to the maximally localized state \eqref{mls}, for which $\Delta \fat x=1/\Delta \fat p =\beta $. The tradeoff curve (solid red line) $\gamma_{\mc U}(\lambda)$ branch below this point corresponds to ground states of (deformed) harmonic oscillators \eqref{optstates} with, from right to left, increasing frequency $\omega \in \,]0,\infty[$. The upper part of the curve (short dashed red line) is obtained by considering states $\psi(k) \propto \cos(\sqrt{\beta} k)^{\gamma_\lambda}$ with $\gamma_\lambda <1$, which are not the ground state of a harmonic oscillator.}  \label{dxdp}
\end{figure}
Uncertainty relations express the fact that there is no state such that both variances are getting simultaneously arbitrary small.
If this is the case, the point $(0,0)$ is not contained in $\mathcal{U}$ and the uncertainty diagram has some empty space around the origin.

However, we can still ask for the ``smallest'' points in $\mathcal{U}$, i. e. the points on the ``lower left'' boundary of $\mathcal{U}$, which are obtained by minimising one variance under the constraint that the other stays below some fixed threshold, and vice versa \citep{dammeier2015uncertainty, abdelkhalek2015optimality}. In Fig.\ref{dxdp} this trade-off curve is indicated by a solid red line and henceforth referred to as an {\it optimal} and {\it state-independent} uncertainty relation.
Here, the term ``optimal'' means that for any attainable value for $\Delta\fat p$, or equivalently for $\Delta\fat x$, we can find a state in $\mc P$ such that the uncertainty relation is tight, i.e. that equality is attained.
``State-independent'' means that the uncertainty relation only depends on functions of the variances $\Delta\fat p$, $\Delta\fat x$ and constants, but not on any other quantities that depend on the state.
Hence an optimal and state-independent uncertainty relation defines a trade-off curve such that for any attainable value of $\Delta \p$, we can directly conclude the {\it best} lower bound on $\Delta \x$ that will hold for {\it all} states in $\mathcal{P}$, and vice versa.

Importantly, minimal length as defined in \eqref{minlengthdef} can be directly computed if such an uncertainty relation is known by simply optimising over all possible values of $\Delta \p$.
For this and other operationally motivated reasons pointed out by David Deutsch in the 80's \cite{deutsch1983uncertainty}, optimal and state-independent uncertainty relations are the ones to look for. However, such uncertainty relations are usually also the hardest ones to obtain because they always involve a constrained optimization problem over the whole state space.

At this point it might be important to recall the often ignored fact that, in general, an optimal and state-independent uncertainty relation cannot be inferred from the relation \eqref{rob}, which in our case takes the form
\begin{align}
\label{robxp}
\Delta \x \Delta \p \geq \frac{1}{4}|\bra{\psi} f(\p) \ket{\psi} |^2 \ .
\end{align}
Here, the expectation value on the right hand side of \eqref{robxp} (or \eqref{rob}) is generally {\it state-dependent}, such that evaluating the uncertainty relation for a particular state does not allow to directly conclude anything about the uncertainty of any other state. In particular, it is generally not true, in neither direction, that states, which are giving equality in \eqref{rob}, correspond to points on the boundary of an uncertainty region. This circumstance can be exemplary checked by considering any non-commuting pair of measurements, e.g. two angular momentum components \cite{dammeier2015uncertainty}.  This has also been pointed out by \cite{spindel, dorsch} for the case of the smallest value of $\Delta \x$. In general, this makes the method to investigate equality in \eqref{rob} and to infer minimal length, quite problematic.

However, there are at least two exceptions to the above criticism:
 the one is the usual Heisenberg algebra, where the right hand side of \eqref{rob} is the same for every normalized state and thus state-independent. Optimality is granted by Gaussian states, for which it is well-known that they achieve equality in \eqref{rob} in the whole parameter range of $\Delta \p$ and $\Delta \x$, respectively.

The other exception has been exploited in \citep{kempf1995hilbert} for the case of a modified Heisenberg algebra with $f(p)=1+\beta p^2$. If we take the square root on both sides of \eqref{robxp}, the right hand side only contains a constant, some factors and the second moment of $\p$ and thus one obtains a state-independent uncertainty relation
\begin{align}
\label{ucr95}
\sqrt{\Delta \x \Delta \p}  &\geq \frac{1}{2}\left( 1+\beta \Delta \p+\langle\fat p\rangle^2\right)\geq\frac{1}{2}\left( 1+\beta \Delta \p\right).
\end{align}
Whilst this was not spelled out in \citep{kempf1995hilbert} directly, this bound is in fact optimal as we will show at the end of this section by a straightforward application of Theorem~\ref{thm:convex}.

For a large class of modifications (see Subsection~\ref{nonoptbound} in the appendix) we can set
\begin{align}
g(p):=f(-\sqrt{|p|})\ ,
\end{align}
and substitute this into \eqref{robxp}. Then, in a similar spirit as above, we obtain
\begin{align}
\label{ucrbla}
  \Delta \x &\geq \frac{g(\Delta \p)^2}{4 \Delta \p} \ ,
\end{align}
which is state-independent but in general not optimal (see the blue line in Fig.~\ref{ucrbounds}).

\begin{figure*}
  		\hspace{0.4cm}
  		\centering \def\svgwidth{0.35\textwidth}
  		\input{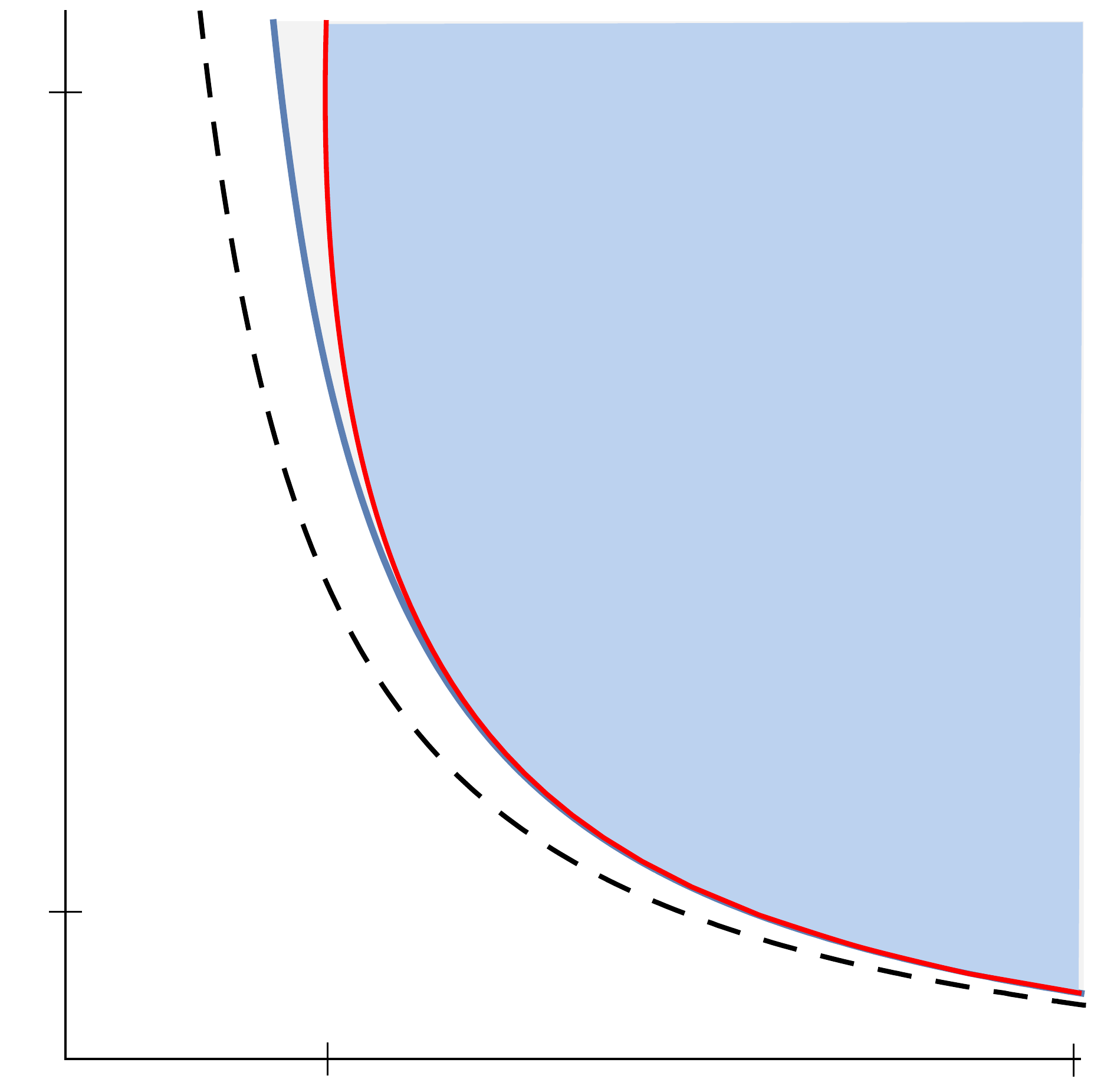}
  		\hspace{2.4cm}
  		\def\svgwidth{0.35\textwidth}
  		\input{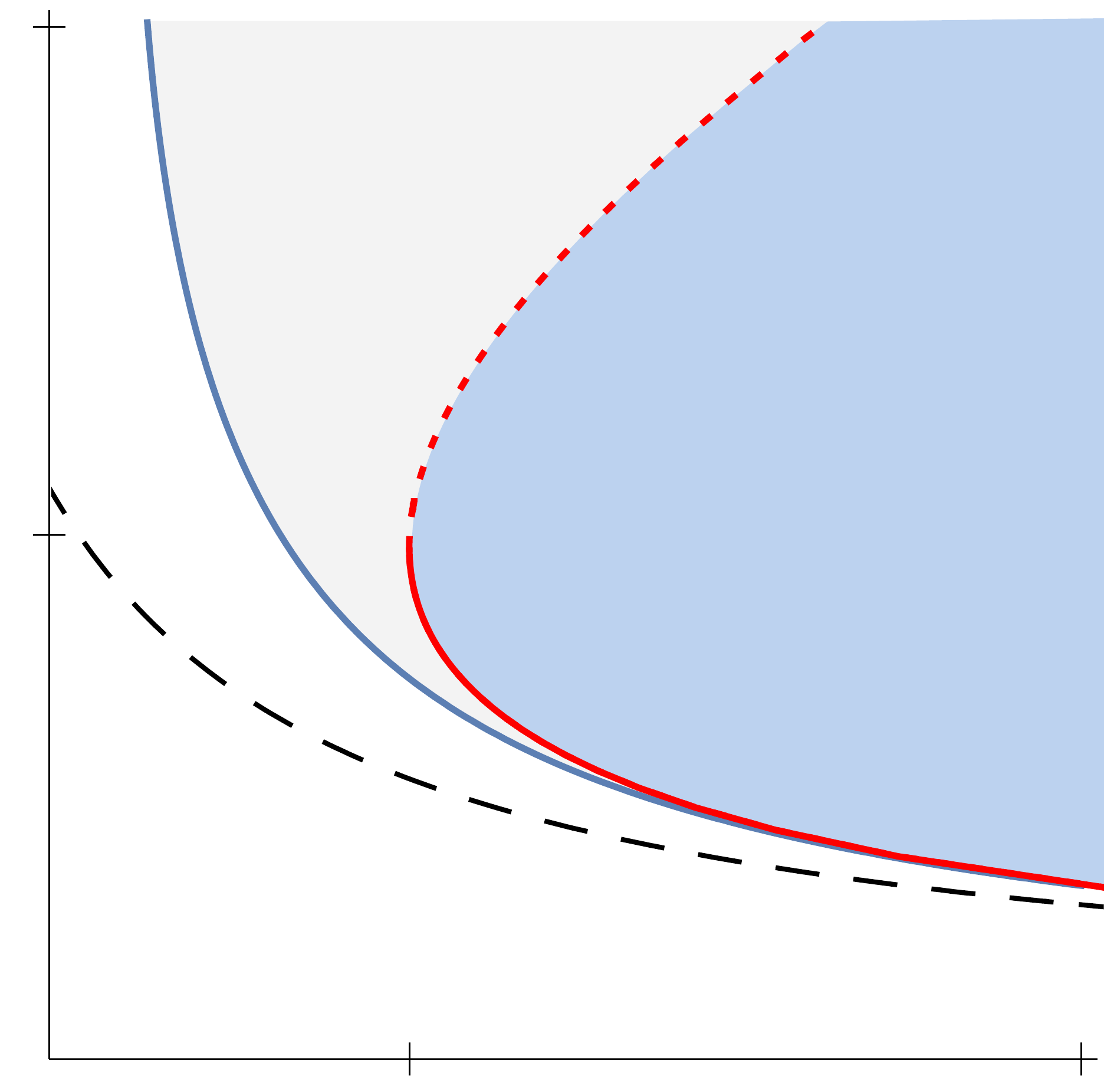}
  		\caption{
  			Uncertainty regions $\mc U$ for the modifications $f_1=\cosh(\sqrt{\beta} p)$ (left figure) and $f_2=1+\beta p^2+ \beta^2 p^4/4$ (right figure) bounded by the red lines. The black dashed lines show the standard Heisenberg bound for the position and unmodified momentum operator. The blue lines show the not optimal but state-independent bounds from \eqref{ucrbla}.}

  		\label{ucrbounds}
\end{figure*}

To the best of our knowledge no universal method for obtaining an optimal, state-independent uncertainty relation for an arbitrary pair of observables $A$ and $B$ is known. However, it is possible to obtain lower bounds (and by this a state-independent uncertainty relation) on every uncertainty region by computing its convex hull (see \citep{dammeier2015uncertainty,convexopti,srw} and Subsection~\ref{robbound} in the appendix). Such a bound will become optimal whenever $\mathcal{U}$ itself is convex. In this case an uncertainty relation can always be characterised by a function $u(\lambda)$ with $\lambda\in[0,1]$ and a set of linear inequalities
\begin{align}
\label{linbound}
\lambda \Delta \x + (1-\lambda)\Delta \p \geq u(\lambda)\ .
\end{align}
The function $u(\lambda)$ will give us a full description of the boundary of the convex hull of $\mc U$ (see Subsection~\ref{robbound} in the appendix and Fig.\ref{fig:convex}). If needed, one can recover the trade-off curve, denoted by $\xi_{\mc U}(\lambda)$ in the following, by the formula
\begin{align}
\label{gammapoints}
\xi_{\mc U}(\lambda)=(u(\lambda)+(1-\lambda)u'(\lambda),u(\lambda)-\lambda u'(\lambda))\ .
\end{align}
Note that, given a particular form of $u(\lambda)$, one can always find a substitution for $\lambda$ in \eqref{gammapoints}, such that \eqref{gammapoints} has a form that only depends on $\Delta \x$ and $\Delta \p$.

The following theorem states that the ansatz above is already sufficient for providing an {\it optimal} uncertainty relation.

\begin{theorem}
\label{thm:convex}
Let $\mathcal{U}$ be the uncertainty region of  $\x$ and $\p$ satisfying a modified algebra
with a modification $f(p)$ that obeys the assumptions described in section~\ref{heisenalg}. Then
\begin{itemize}
\item[1.] the lower boundary, i.e. the trade-off curve, of $\mathcal{U}$ lies completely on the boundary of a convex set,
\item[2.] states corresponding to this trade-off curve always have expectation $\langle \p\rangle=0$ and can be chosen to have $\langle \x\rangle=0$,
\item[3.] these states are ground states of the modified harmonic oscillator $$H_{\lambda}=\lambda\x^2+(1-\lambda)\p^2.$$
\end{itemize}
\end{theorem}
A proof and a mathematically more dedicated formulation of the statements 1 and 2 from Theorem \ref{thm:convex}
can be found in Subsection~\ref{robbound} in the appendix. However, statement 3 can be concluded directly from 1 and 2 using \eqref{linbound}: from 1 we know that we can obtain the optimal bound $u(\lambda)$ by minimising the expression $\lambda \Delta \x + (1-\lambda)\Delta \p$ for fixed $\lambda$ over all states in $\mathcal{P}$.
Using 2 we arrive at
\begin{align}
\label{h01}
u(\lambda)=\min_{\psi\in\mathcal{P}} \bra{\psi}(\lambda \x^2 + (1-\lambda) \p^2)\ket{\psi},
\end{align}
which is exactly the ground state energy of a harmonic oscillator in the modified algebra.
Moreover, when we represent $\p$ and $\x$ in the domain of $\k$, we can state the following:

\begin{cor}
\label{cor:method}
An optimal and state-independent uncertainty relation can be directly obtained by solving the ground state problem of the Schrödinger operator
\begin{align}
\label{schrödinger}
H_\lambda=-\lambda \partial_k^2 + (1-\lambda) p(k)^2\ ,
\end{align}
with Dirichlet boundary conditions at $\pm k_{max}$ and a symmetric, convex and positive potential $p(k)^2$. This uncertainty relation saturates \eqref{linbound} and can be found by
solving
\begin{align}
H_\lambda \psi_\lambda(k)=u(\lambda) \psi_\lambda(k)\ ,
\end{align}
where $u(\lambda)$
is given by the ground state energy of \eqref{schrödinger}.
\end{cor}

Fortunately, these kind of problems have been the content of many extensive studies (see for example \cite{ReSi:1970, teller}) since the early days of quantum mechanics. Indeed, well-established numerical methods and analytical solutions for several particular instances of $p(k)$ are available.

Asking for an optimal bound in expression \eqref{linbound} for the special case $\lambda=1$ directly translates into characterising the minimal length $l^2_{\min}$ in terms of variances. By Corollary \ref{cor:method} this turns into the task of finding the state $\psi(k)$ and the minimal value $u(1)=l_{min}^2$ such that the differential equation
\begin{align}
-\partial_k^2 \psi(k)=l_{\min}^2\psi(k)\ ,
\end{align}
holds with the boundary condition $\psi(\pm k_{\max})=0$. But this is just the ground state problem of a particle in a box with length $2k_{\max}$ and this is solved  by
\begin{align}
\label{mls}
\psi(k)= \frac{1}{\sqrt{k_{\max}}} \cos \left( \frac{\pi}{2} \frac{ k}{k_{\max}} \right)\ ,
\end{align}
so that
\begin{align}\label{minlength}
l_{\min}^2=\frac{\pi^2}{4 k_{\max}^2}.
\end{align}
Note that all these results are completely general and hold for all modifications that satisfy the conditions (i) - (iii) of section~\ref{heisenalg}. As such they not only generalise previous results obtained in \cite{spindel,dorsch}, but also allow to drastically improve earlier approaches as they provide a means to straightforwardly compute optimal uncertainty relations and minimal lengths.

To illustrate this point, let us consider the most studied modification
\begin{equation}
f(\p)=1+\beta \p^2,
\end{equation}
for which a value for minimal length and its corresponding quantum state is already known, while the uniqueness of the construction of $\x$ and $\p$ has been left open. In \cite{kempf1995hilbert} the authors also provide a state-independent uncertainty relation. Using our results we can directly prove that this uncertainty relation is actually optimal. Additionally, we show how, from such an uncertainty relation, one can easily retrieve the aforementioned results which were previously obtained in a much more mathematically involved manner. The purpose of this example is therefore to show the validity of our results and to provide a step-by-step recipe to compute uncertainty relations and minimal lengths by making use of the main results presented so far in this paper.\\
As a first step, note that $f(\p)$ satisfies the requirements (i) - (iii). Hence we know by Theorem \ref{thm:mainthm} that the modified momentum operator $\p$ must be a hermitian operator that satisfies the differential equation \eqref{fkeq}. This yields
\begin{equation}
p(k)=\frac{1}{\sqrt{\beta}} \tan(\sqrt{\beta} k)\ .
\end{equation}
Also, by \eqref{eq:kmax} we can compute the momentum cut-off,
\begin{equation}
k_{\max}=\frac{\pi}{2\sqrt{\beta}} \ .
\end{equation}
The optimal state-independent uncertainty relation is characterised by the trade-off curve of the uncertainty region $\mathcal{U}$, the exact form of which depends on the modification. By Theorem \ref{thm:convex} we know that the states $\psi_{\lambda}$ parametrising this trade-off curve are ground states of the modified harmonic oscillator with Hamiltonian $\H_{\lambda}=\lambda \x^2+(1-\lambda)\p^2$, i.e.
\begin{equation}
\H_{\lambda}\psi_{\lambda}
=u(\lambda) \psi_{\lambda} \ .
\end{equation}
We rewrite this condition by explicitly inserting the parameter $\beta$ to render all terms adimensional. By dividing by $\lambda$ we then have
\begin{equation}
 \left(\frac1\beta \,\x^2 + \frac{1-\lambda}\lambda \beta\,  \p^2 \right) \psi_\lambda = \frac{u(\lambda)}\lambda \psi_\lambda \equiv \gamma(\lambda) \psi_\lambda\ .
\end{equation}
The ground states of these Hamiltonians correspond to vectors in the kernel of the annihilation operator
\begin{equation}
\aa_{\lambda}=\x/\sqrt{\beta}  + i \gamma_{\lambda} \sqrt{\beta}\, \p\ ,
\end{equation}
since $\H_{\lambda}$ always\footnote{Note that this ansatz \emph{only} works for a modification of a form $f(p)=\fat{1}+\beta \p^2$} satisfies
\begin{equation}
\frac{1}{\lambda}\H_{\lambda}= \frac1\beta \x^2+  (\gamma_{\lambda}^2 - \gamma_{\lambda}) \beta \,\p^2 = \aa_{\lambda}^{\dagger} \aa_{\lambda} + \gamma_{\lambda} \mathbb{I}\ ,
\end{equation}
when choosing $\gamma_{\lambda}$ such that $\gamma_{\lambda}^2-\gamma_{\lambda}=(1-\lambda)/\lambda$. Hence, we have $\aa_{\lambda} \psi_{\lambda}=0$, which translates into the differential equation
\begin{equation}
\partial_k \psi_{\lambda}(k) + \gamma_{\lambda}\sqrt{\beta} \tan(\sqrt{\beta} k) \psi_{\lambda}(k)=0\ ,
\end{equation}
with the solution
\begin{equation}\label{optstates}
\psi_{\lambda}(k) = \left( \frac\beta\pi \right)^{1/4} \left( \frac{\Gamma(1+\gamma_\lambda)}{\Gamma(1/2+\gamma_\lambda)} \right)^{1/2} \cos(\sqrt{\beta}k)^{\gamma_{\lambda}} \ ,
\end{equation}
where
\begin{equation}
\gamma_\lambda = \frac12 \left( 1+ \sqrt{1+4\frac{1-\lambda}{\lambda}} \right)\ .
\label{res-min}
\end{equation}
These states parametrise the complete trade-off curve of the uncertainty region $\mathcal{U}$ as depicted in Fig. \ref{dxdp} and therefore yield an optimal state-independent uncertainty relation for the modification $f(\p)=1+\beta \p^2$. More concretely, we can evaluate the variances $\Delta \x$ and $\Delta \p$ for these states
\begin{equation}
\Delta \x = \beta \frac{\gamma_\lambda^2}{2 \gamma_\lambda-1}\ , \,\,\,\,\, \Delta \p= \frac{1}{\beta}\frac{1}{2 \gamma_\lambda-1}\ ,\label{meanvalues}
\end{equation}
where we invoked Theorem \ref{thm:convex} to set $\langle \x\rangle=\langle \p\rangle=0$.

Notice that for $\lambda \rightarrow 0$ (i.e. $\gamma_\lambda \rightarrow \infty$), the state becomes a plane wave, while $\lambda=1$ ($\gamma_\lambda =1$) corresponds to
the maximally localized state
\begin{equation}\label{specms}
\psi(k)=\sqrt{\frac{2\sqrt{\beta}}{\pi}} \cos ( \sqrt{\beta} k ) \ ,
\end{equation}
for which $\Delta \x = l_{\min}^2 = \beta$, compare with \eqref{mls} and \eqref{minlength}.
The results \eqref{optstates} and \eqref{res-min} as well as the trade-off curve coincide with those obtained in \cite{kempf1995hilbert} (compare with Eq. (69) in that paper) with the identification
\begin{equation}
\frac{1-\lambda}{\lambda} = \frac{1}{(\beta m \omega)^2}\ . \label{lambdaomega}
\end{equation}
However, our findings greatly simplify and extend the derivation of these results, while proving uniqueness properties and the optimality of the state-independent uncertainty relation, and allowing to treat any modification $f$ that satisfies (I) - (III).

The above discussion allows for interesting physics to become visible: when looking at the trade-off curve traced out by the states $\psi(k)$ (see Fig.~\ref{dxdp}) the position variance decreases with increasing momentum variance - exactly up to the point where the frequency of the harmonic oscillator diverges, $\lambda=\gamma_\lambda=1$. At this point the state reaches the maximal possible localisation in space, the endpoint of the solid red line in Fig. \ref{dxdp}. Actually, using (\ref{meanvalues}) it is easy to check that the states (\ref{optstates}) still saturate the generalized uncertainty principle bound even for $\gamma_\lambda<1$ but they {\it do not correspond to the ground state of a harmonic oscillator}, see (\ref{lambdaomega}), but rather can be formally seen as eigenstates of a quadratic potential with {\it an imaginary frequency $\omega$}. This regime corresponds to the upper branch in Fig. \ref{dxdp} (dashed red line). When $\gamma_\lambda$ decreases, both $\Delta \x$ and $\Delta \p$ grow and diverge in the limit $\gamma_\lambda \rightarrow 1/2$. Yet, the states with any $\gamma_\lambda > -1/2$ are normalizable, so that we can associate to them an entropy in both momentum and position space, as we will see in the next section.

\section{Entropic Bounds}\label{sec:entrops}


In this section we introduce, compute and discuss implications of an entropic formulation of uncertainty and minimal length. Here, one might be tempted to ask why using variances is not always a good choice to quantify the uncertainty of two measurements, especially since we dedicated the whole last section to exactly this setting. The answer is that the emphasis of the previous section lies in the formulation of optimal and state-independent uncertainty relations which best describe minimal uncertainties and are always superior to statements about minimal length only or state-dependent uncertainty relations. The concept of optimal and state-independent uncertainty relations is however, completely independent of the chosen uncertainty measure: one can formulate such relations using variances as done in the last section, or  compute so-called entropic uncertainty relations as suggested by David Deutsch in his seminal paper \cite{deutsch1983uncertainty}, which will be the content of this section. Before introducing entropies, let us first clarify why variances as measure of uncertainty are problematic.

In \cite{deutsch1983uncertainty} David Deutsch argued that variances suffer from the fact that they depend on the specific ordering and labeling of measurement outcomes. To illustrate this point consider a fair coin that yields ``heads'' or ``tails'' with equal probability. To be able to quantify the uncertainty about the outcome of one toin coss in terms of variances, one needs to artificially associate real numbers to heads or tails. In other words, variances depend on the choice of the labels of the possible outcomes - to the extent that we can choose this ``measure of uncertainty'' to become arbitrarily small or large, while intuitively our uncertainty is the same, independent of the labeling.

As another example that is related to this problem let us consider a spin measurement on a spin-1 particle \cite{coles2015entropic, bialynicki2011entropic}.
Let us assume that we only know that the outcomes $\{-1,+1\}$ occur with equal probability $p_{-1}=p_{+1}=1/4$, whereas with highest probability $p_0=1/2$ we obtain outcome 0. Now imagine that we get {\it additional information} about the source telling us that we never obtain outcome zero. Our state of knowledge changes and so does the probability distribution, which now is given by $p_{-1}=p_{+1}=1/2$ and $p_0=0$. Here a ``good'' measure of uncertainty should mirror our decrease of uncertainty by not increasing during this process. However, variances do not satisfy this minimal requirement: in fact, the variance in the above example will increase by getting further information and it is easy to construct similar examples when we consider continuous observables as well.

For unbounded observables yet another problem arises: namely the variance of a random variable can diverge although the corresponding probability distribution seems to be ``located'' in some sense. Prominent examples for this are Cauchy distributions and L\'{e}vy distributions. Moreover, this effect also occurs for two of the distributions shown in Fig.~\ref{posprob}. Here, the black and the red curve correspond to distributions which appear to be ``localised'' even though their variances diverge.

All these examples are a consequence of the variance depending on the outcomes and not only on the underlying probability distribution.
In finite dimensions a well-known alternative to variances, that does not suffer from this drawbacks, are entropies \cite{deutsch1983uncertainty,WehnerWinter,maassen1988generalized}. The most prominent one is the Shannon entropy\footnote{In this work ``$\log $'' always refers to the natural logarithm.} of a discrete probability distribution $w:\mathbb{Z}\rightarrow(0,1)$,
\begin{align}
\label{shannondiscrete}
H(w):=-\sum_i w_i \log (w_i) \ ,
\end{align}
 which was introduced in the seminal work \cite{shannon}.
Later Alfréd Rényi introduced a whole family of entropies $H_{\alpha}$ \cite{renyi1961}, called Rényi-$\alpha$ entropies, which also do not suffer from the above drawbacks and contain the Shannon entropy in the limiting case $\alpha \rightarrow 1$.
In this work we will consider only the Shannon entropy and the \textit{min-entropy} $H_\infty$, which arises in the limit $\alpha \rightarrow \infty$
\begin{align}
\label{minentro}
H_\infty (w)=-\log(\max_i w_i) \ .
\end{align}
Note that these two entropies are so far only defined in a finite-dimensional setting. In the following we will define and compute minimal length in terms of Shannon entropies and min-entropies for continuous variables.

\subsection{Shannon entropy}
In \cite{shannon} Shannon presented a generalisation of \eqref{shannondiscrete} for continuous variables with a probability density $w:\mathbb{R}\rightarrow\mathbb{R}$
\begin{align}
h(w)=-\int dy \; w(y)\log\left(w(y)\right),
\end{align}
which is called the \textit{differential entropy}. This quantity can be negative and even reach the value $-\infty$, which might, on a first view, appear to be an astonishing property for an uncertainty measure. We therefore give some clarification of its meaning. Consider a continuous valued observable given by a random variable $Y$ with outcomes $y$ on the whole real line. Now assume that an experimenter tries to measure this observable with a device that has a finite operating range, lets say in an interval $I$. Assume further that her measurement device only has a finite resolution, say $\varepsilon$, which means that the device can only decide whether the outcome of a measurement is in a particular interval of length $\varepsilon$ or not. Now the experimenter can divide the operating range $I$ into bins $\Omega_\varepsilon^i$ of length $\varepsilon$ and will thus, effectively, obtain a description of her measurement by a discrete random variable, say $Y_\varepsilon^I$ with approximately $|I|/\varepsilon$ different outcomes. Here computing entropies like \eqref{shannondiscrete} or \eqref{minentro} for this random variable will give her a good description of the information theoretic uncertainty.

The experimenter might then take a better device, i. e. one with a finer resolution and a larger operating range. In this case the entropy increases because the number of possible measurement results will increase. Moreover, in the limit $|I|\rightarrow\infty$ and $\varepsilon\rightarrow 0$ the entropy will reach infinity. Nevertheless, assuming that $Y$ is distributed by $w$, for \eqref{shannondiscrete}, we can write down this limit as
\begin{align}
\lim_{\substack{|I|\rightarrow\infty\\\varepsilon\rightarrow 0}} H(Y_\varepsilon^I)=
\lim_{\substack{|I|\rightarrow\infty\\\varepsilon\rightarrow 0}} \hspace{-.5em} -\sum_i \mathbb{P}(Y,\Omega_\varepsilon^i)\log(\mathbb{P}(Y,\Omega_\varepsilon^i)) \ ,
\end{align}
where $\mathbb{P}(Y,\Omega_\varepsilon^i)$ denotes the probability of measuring a result in the bin $\Omega_\varepsilon^i$. For small $\varepsilon$ and a bin with center $y_i$ we might approximate this probability by $\varepsilon w(y_i)$ and get
\begin{eqnarray}
H(Y_\varepsilon^I)& \approx&-\sum_i \varepsilon w(y_i) \log\left( w(y_i) \right) \nonumber \\
&-&\sum_i \varepsilon w(y_i) \log(\varepsilon)\ ,
\end{eqnarray}
which gives
\begin{align}
\label{limit}
\lim_{\varepsilon\rightarrow 0}\lim_{|I|\rightarrow\infty} H(Y_\varepsilon^I)=h(w)+\infty \ .
\end{align}
Hence the quantity $h(w)$ can be understood as a deviation from infinity. However the limit, in \eqref{limit}, strongly depends on the experimenter choice of dividing the interval $I$ into equidistant bins. This choice corresponds to the assumption that, not knowing anything about $Y$, the probability of obtaining an outcome in any bin $\Omega_\varepsilon^I$ should be the same when sampling from a uniform distribution on $I$. This assumption might become controversial (see \cite{Jaynes63}) when taking the limit $|I|$ to infinity, because there is no notion of a uniform distribution on the whole real line. In the following, we will see that we will have to take care of such a choice of a reference measure when defining entropies for the modified momentum.

 At this point we should also emphasise that, even if the absolute value $h(w)$ has only a rather indirect operational meaning,  $h(w)$ is still a good quantity to judge if one distribution is ``sharper'' than another, and thus minimizing $h(w)$ will give us a good notion to characterise minimal length quantum states. Moreover, $h(w)$ can be used to compute other information theoretic quantities like the mutual information $I(A,B)=h(A)+h(B)-h(A B)$ which quantifies the correlation between two random variables $A$ and $B$.
 Here $h(A)$ denotes the Shannon entropy of the probability density of the random variable $A$.
 It was shown in \cite{shannon}, that the ``$\infty$'' term from \eqref{limit} cancels out, such that  $I(A,B)$ arises as a rigorous limit of a discrete quantity.

Let us now define and compute the corresponding minimal length in terms of the Shannon entropy. To this end, we consider
the Shannon entropy of a probability distribution obtained by measuring $\x$ on a state $\psi\in\mathcal{P}$. Assuming that we are given $\psi$ as a function $\psi(k)$ of the coordinate $k$ we can obtain its representation as a function $\phi(x)$ of the coordinate $x$ by applying a Fourier
transformation. In this case the amplitude $|\phi(x)|^2$ will correspond to a probability density on $\mathbb{R}$ normalized with respect to the measure '$dx$' and we set
\begin{align}
h_\x(\phi)=-\int_\mathbb{R} dx |\phi(x)|^2 \log |\phi(x)|^2\ ,
\end{align}
in order to define the Shannon entropy of a position measurement.
In analogy to the ``standard'' (but problematic) definition of minimal length, minimal length {\it in terms of entropies} is now defined by the minimal entropy that a probability distribution obtained by a position measurement can have \footnote{Clearly, entropy does not have the unit of a length. However, the exponential of the entropy does (see \cite{Hall} where it is shown that the exponential of entropies provide a natural way to assign ``length'' to a quantum state). For simplicity, we drop this exponential in this work.}
\begin{align}
\Gamma_{\min}=\min_{\psi\in \mathcal{P}} h_\x(\mathcal{F}(\psi))\ .
\end{align}

Measuring the unmodified momentum $\k$ on $\psi\in\mathcal{P}$ will give us the probability distribution $|\psi(k)|^2$ such that we will define the unmodified momentum entropy as
\begin{align}
h_\k(\psi)=-\int_I dk |\psi(k)|^2 \log |\psi(k)|^2\ ,
\end{align}
where the interval $I$ ranges from $-k_{\max}$ to $+k_{\max}$.

\begin{figure}
	\def\svgwidth{0.45\textwidth}
	\centering
	\input{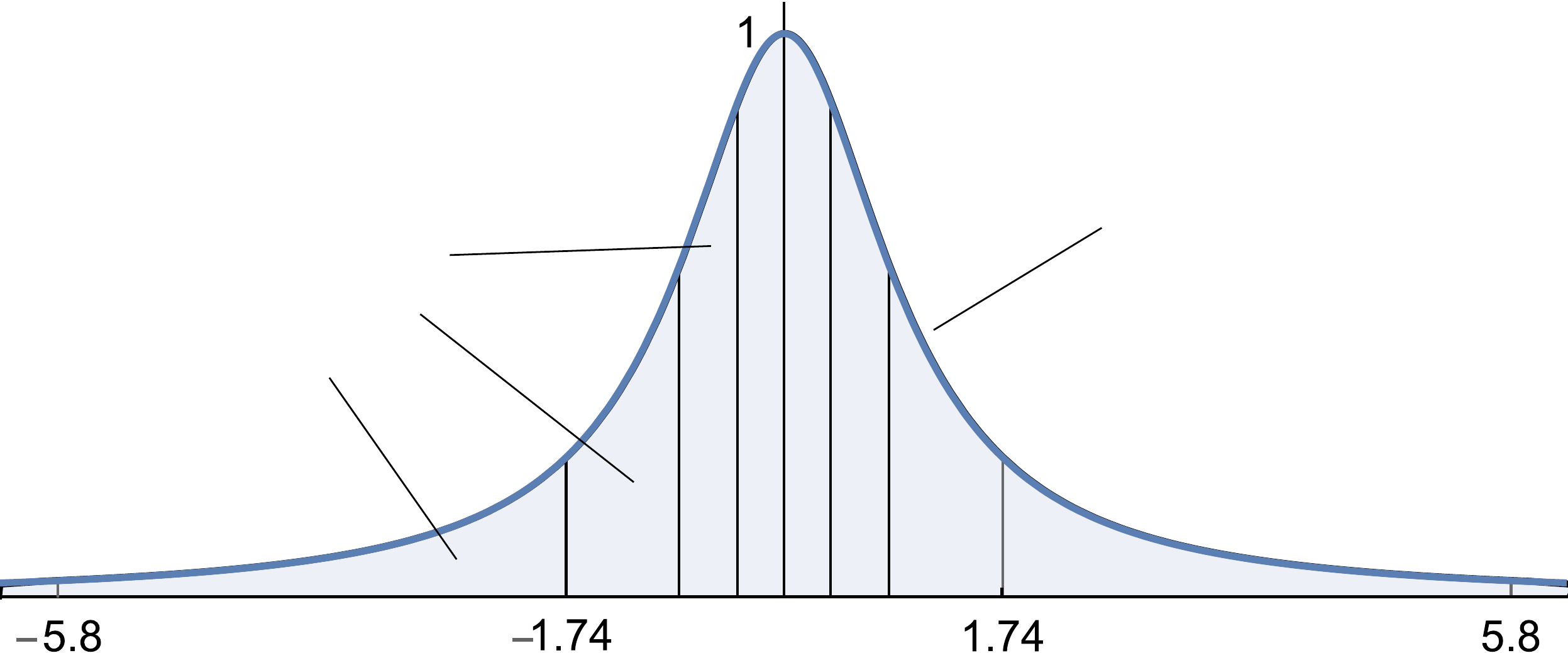}
	\caption{ Different binnings of measurement data lead to different uncertainty measures. Here we choose a binning such that all bins have the same probability in our definition of the Shannon entropy.}
	\label{bins}
\end{figure}

As mentioned in the previous subsection, the subtlety of choosing an appropriate reference measure emerges: when we represent $\psi$ as a function $\tilde{\psi} (p)$ of the coordinate $p$, the probability of measuring a certain value $p$ from an interval $(a,b)$ is given by
\begin{align}
\mathbb{P}(p\in(a,b),\p,\psi)=\int_a^b \frac{dp}{f(p)}|\tilde{\psi}(p)|^2 \ ,
\end{align}
which is no longer translation invariant, due to the scaling factor $1/f(p)$.
Here the experimenter from the above example has to adapt to this when choosing bins: one choice would be to keep on taking bins with equal length.
Another choice, the one we use in this work, is to take bins such that for all bins the probability obtained by measuring a function $\tilde{\psi}(p)$, which is constant on a particular bin, is the same (see Fig.~\ref{bins}). In this case a bin, with center $a$ and volume  $\varepsilon$, will correspond to an interval $(p(-\varepsilon/2+p^{-1}(a)),p(p^{-1}(a)+\varepsilon/2)$, where $p$ and $p^{-1}$ are obtained by representing $\p$ as a function of $\k$. By evaluating the limit of $\varepsilon\rightarrow 0$ we see that we therefore should define the entropy of a modified momentum measurement via\footnote{We also note that the measure $\mu(dp) = dp/f(p)$ guarantees that in ${\mathcal H}={\mathcal L}^2(\mathbb{R},\mu(dp))$ the $\x$ operator can be represented as $\x := i f(p) d/dp$. Equivalently, in the $k$ representation the measure is simply $dk$ and $\x := i d/dk$, see \cite{kempf1995hilbert}.}
\begin{align}
\label{hp}
h_\p(\tilde{\psi})=-\int_\mathbb{R}\frac{dp}{f(p)}|\tilde{\psi}(p)|^2 \log |\tilde{\psi}(p)|^2 \ .
\end{align}
Notice that for $\p$ and $\k$ both choices will lead to the same definition of an entropy. Furthermore,  $h_\p$ as defined above has the nice advantage that it arises from $h_\k$ by an integral substitution and thus does not change, i. e.
\begin{align}
h_\p(\tilde{\psi})=h_\k(\psi)\ .
\end{align}
Thus, the optimization of $h_\p$ over all states represented as functions of $p\in\mathbb{R}$ amounts to optimizing
$h_\k$ over all functions of $k\in I$.
Finally, $h_\p$ also depends on $f(p)$ only through the cut-off parameter $k_{\max}$. All results that can be shown for an arbitrary $k_{\max}$ are therefore valid for arbitrary modifications $f(p)$.

Having defined the entropic uncertainty measures, we can now compute the corresponding uncertainty relations. To this end, we consider all possible pairs $(h_\x(\phi),h_\k(\psi))$ which are attainable by $\psi(k)\in\mathcal{P}$ where $\phi(x)=\mathcal{F}[\psi](x)$.  We recall that in the standard scenario, $\beta=0$, both momentum and position entropies can become arbitrarily small or large if one consider sequences of states which converge (in the distribution sense) to the $\x$ or $\p$ eigenfunctions. In this case the ``physical'' states satisfies the Bialynicki-Birula (BB) bound \cite{bialynicki1975uncertainty,Bialynicki-B:1984}, obtained by the Babenko-Beckner inequality \cite{beckner1975inequalities},
\begin{align}
h_\k(\psi)+ h_\x(\phi) \geq \log( \pi e)\ ,
\end{align}
 which is again saturated, as for the product of variances (see Theorem \ref{thm:convex}), by the ground states of harmonic oscillators with arbitrary frequency $\omega$.

 For a modified algebra, the existence of a momentum cut-off is expected to imply a minimal value for $h_\x$. Before  discussing its value, it is worth noticing that the representation of a modified algebra also implies a {\it maximal entropy $h_p$ in momentum space}. This maximal entropy is attained for any series of functions converging to the uniform distribution on $I=[-k_{\mathrm{max}},k_{\mathrm{max}}]$, and reads
\begin{align}
\label{hpmax}
\max_{\psi\in\mathcal{P}}h_\k(\psi)&= -\int_{I}\frac1{|I|}\log\left(\frac1{|I|}\right)dk =  \log(2k_{\max}).
\end{align}
In particular for the family of states from \eqref{optstates} this bound is obtained if we take the limit $\gamma_{\lambda}\rightarrow 0$, (see Fig.~\ref{propk}).

If we combine \eqref{hpmax} with the (BB) bound we find that
\begin{align}
  h_\x(\phi)\geq 1-\log\left(\frac{2k_{\max}}{\pi} \right),
\label{minhx}
\end{align}
implying a lower bound on the entropy $h_\x$.
For a modification $f(\p)=1+\beta\p^2$ this reads
\begin{align}
\label{hpmax2}
h_\k(\psi)\leq\log(\pi)-\frac12 \log(\beta)
\end{align}
and
\begin{align}
  h_\x(\phi)\geq 1+\frac12\log(\beta).
\end{align}
Yet, this bound is not optimal, because the (BB) bound, the dashed line in Fig. \ref{fig1}, becomes tight only on Gaussian functions, see \cite{lieb}. Still however, all entropy pairs have to lie ``above'' this line.

\begin{figure}
\def\svgwidth{0.5\textwidth}
\centering
\input{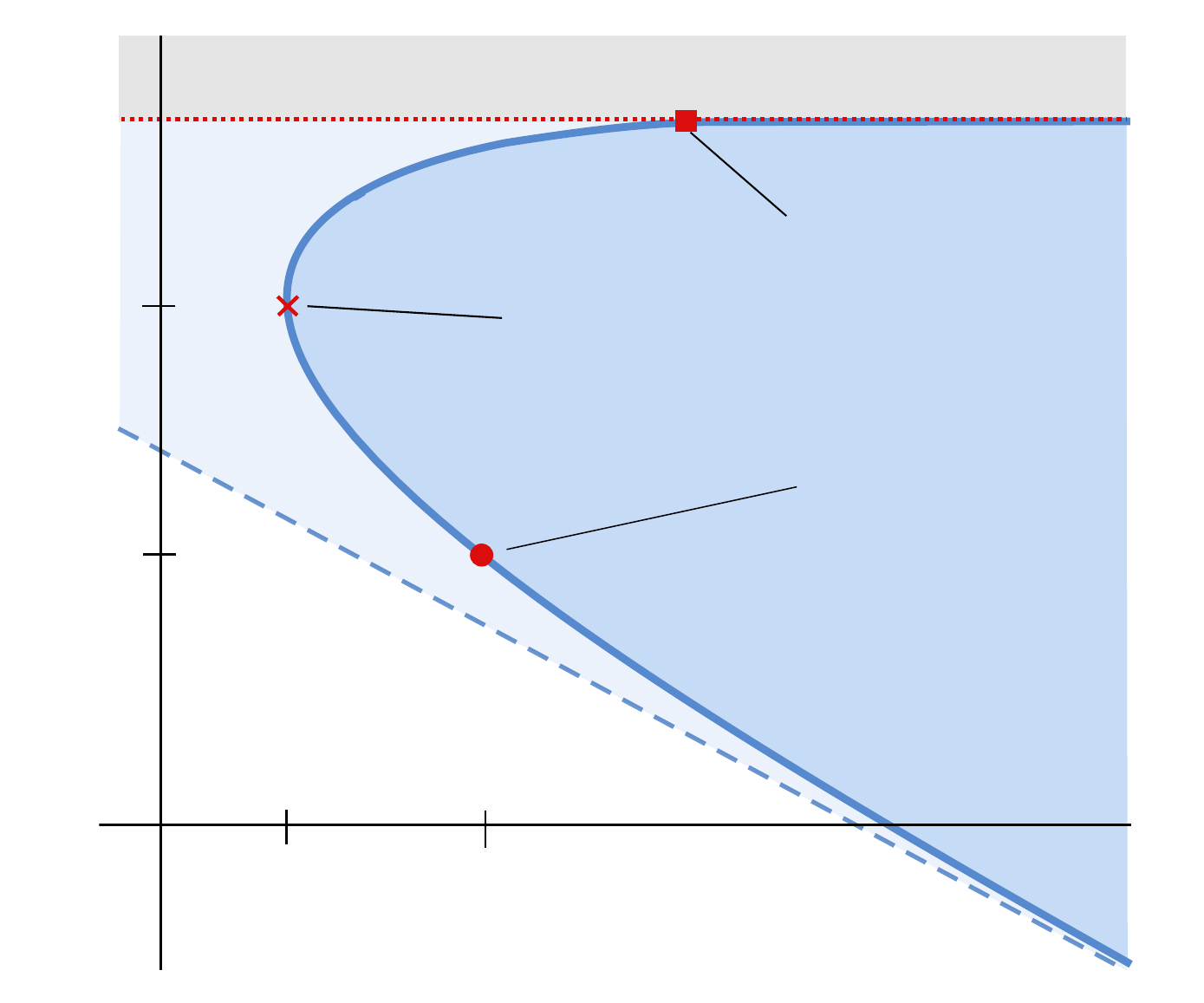}
\caption{Entropic uncertainty region in the $h_\x -h_\p$ plane for the modification $f(\p)=1+ \beta \p^2$.
The dashed line is the Bialynicki-Birula bound which is still valid in our setting but not longer optimal. The horizontal dotted line shows the upper bound on the entropy $h_\p$. The marked points correspond to the $h_\x - h_\p$ pairs for the maximally localized states, the state of minimal position entropy and the one with maximal momentum entropy, respectively. Physical states are conjectured to lie on the right of the solid line and below the dotted horizontal line.}
\label{fig1}
\end{figure}

In analogy with the standard result for the entropy bound, and motivated by our results of the previous section about the optimal uncertainty relation in terms of variances, we conjecture that the analogue of the
(BB) curve corresponds to the states $\psi_{\lambda}(k) \propto \cos(\sqrt{\beta} k)^{\gamma_\lambda}$, which we saw  for $\gamma_\lambda\geq1$ represents the ground state of the deformed harmonic oscillator, while for $\gamma_\lambda<1$ they still saturate the optimal bound in terms of variances, but can be seen as eigenstates of an imaginary frequency oscillator. This curve is shown as the convex solid line in Fig.~\ref{fig1}.

We were unable to obtain an analytic form for the corresponding values of $h_\x$, which have been computed numerically. On the other hand, we can give a simple expression for $h_\k$ in terms of special functions. To this end let us observe that the states (\ref{optstates})
can be also written as
\begin{align}
  \psi_{\lambda}(k) &= \frac1{\sqrt{\kappa(\gamma)}} \exp \left( - \gamma   \,\int_0^k p(k') dk' \right),
\label{psiu}
\end{align}
where we omit the argument $\gamma\equiv \gamma_\lambda$ for readability and
\begin{align}
  \kappa(\gamma) &= \sqrt{\frac \pi \beta} \, \frac{ \Gamma(1/2 + \gamma)}{\Gamma(1+ \gamma)}  .
\label{norm}
\end{align}
Thus, their entropy $h_\k$ reads
\begin{align}
  h_\k &= \int_I dk \frac1{\kappa(\gamma)} \exp \left( - 2 \, \gamma \int_0^k p(k') dk' \right) \times \nonumber \\
    &\times \left(2 \,\gamma \int_0^k p(k') dk'+ \log(\kappa(\gamma)) \right),
\end{align}
or
\begin{align}
  h_\k &= \log(\kappa(\gamma)) - \frac1{\kappa(\gamma)}  \gamma \frac{d}{d\gamma} \kappa(\gamma).
\end{align}
Using (\ref{norm}), we finally find
\begin{align}
  h_\p = h_\k &= \log \left( \sqrt{\frac \pi \beta} \, \frac{ \Gamma(1/2 + \gamma)}{\Gamma(1+ \gamma)} \right) \,+ \nonumber \\
    & + \, \gamma\, {\cal N}(\gamma)-\gamma \,{\cal N}\left( \gamma-\frac12 \right) ,
\end{align}
with ${\cal N}(\gamma)$ the harmonic numbers.

This new boundary curve, shown in Fig. \ref{fig1}, can be divided into three parts, corresponding to different properties of the optimal states (\ref{psiu}) (we fix here $\beta=1$):\\

$\bullet$ for large positive values of $h_\x$ and large negative $h_\p$ the curve asymptotically reaches the (BB) bound, as expected (the role played by the cut-off $\beta$ can be neglected in this regime). As $h_\x$ decreases the curve starts bending and leaves the (BB) straight line. The solid circle shown in Fig.~\ref{fig1} denotes the states of maximal localization in terms of variances studied in the previous section, for which
\begin{align}
h_\k= \log(2 \pi /e)\text{ and } h_\x\simeq 1.374.
\end{align}
\begin{figure}[t]
  \def\svgwidth{0.38\textwidth}
  \centering
  \input{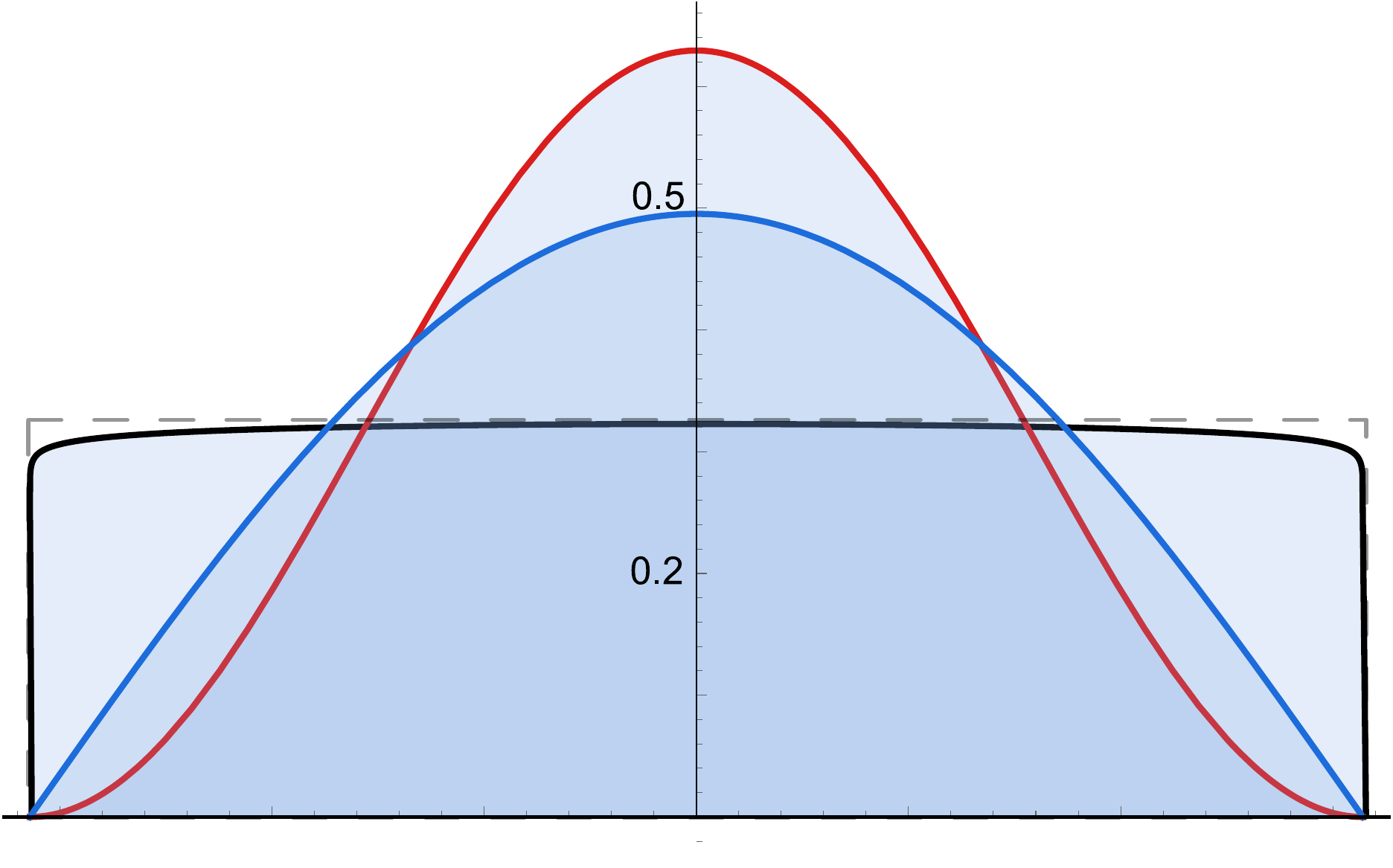}
  \caption{
  	 Minimal length states $\eqref{optstates}$ for $\beta=1$ and $\gamma_{\lambda}=1$ (red, dashed, variance), $\gamma_{\lambda}=1/2$ (blue, dotted, Shannon entropy) and $\gamma_{\lambda}\approx 0$ (black, solid, min-entropy) as in Fig.~\ref{posprob} but in $k$-representation. These minimal length states differ from the states with maximal momentum uncertainty in the case of variances and Shannon entropy. For min-entropies however the minimal length state also yields maximal possible momentum uncertainty.}
  \label{propk}
\end{figure}
Till this point the optimal states are the $\psi_\lambda(k)$ with $\gamma_\lambda\geq1$, i.e. the ground states of the deformed harmonic oscillators;\\

$\bullet$ as $\gamma$ falls below unity, the value of $h_\x$ continues to decrease, until $\gamma=1/2$, which corresponds to the cross in Fig.~\ref{fig1}. We have for this state
\begin{align}
  h_\k= 1\text{ and } h_\x\simeq 1.310 ,
\end{align}
which represents the state of minimal entropy in position.  We see that {\it minimal length in terms of entropy is not equivalent to minimal length in terms of variances}. Indeed, all optimal states in this branch of the curve have finite $\x$ and $\p$ variances, with the exception of the point $\gamma=1/2$, see the previous section. The minimal position entropy can thus, be attained by considering a sequence of such states with $\gamma\rightarrow1/2$. In this limit the variances of both $\x$ and $\p$ diverge (see Fig.~\ref{dxdp}), while their entropies stay finite; \\

$\bullet$ for even smaller values of $\gamma<1/2$, $h_\x$  increases and so does $h_\k$ until it reaches its maximal value, corresponding to a constant wave function in $I$ ($\gamma=0$), up to an arbitrary $k$ dependent phase. This state is shown in the entropy plane as the filled square in Fig.~\ref{fig1}
\begin{align}
  h_\k= \log(\pi)\text{ and }h_\x\simeq 1.524 ,
\end{align}
This part of the curve,  $1/2 > \gamma > 0$, corresponds to normalizable wave functions but with infinite variances for both $\x$ and $\p$. Thus, it is an open boundary for "physical" states, if by so we mean states which are in the domain of position and momentum operators.\\

The solid line and the part of the horizontal line $h_\k= \log(\pi)$ starting from the filled square bounds a convex region. Our conjecture  is that this is in fact, the region in which all entropy pairs for states from $\mathcal{P}$ have to lie in. We cannot present here a proof of this, but we have performed a numerical scan of pairs $h_\x -h_\k$ corresponding to a random sample of $100000$ states built from superposing low excited states (see Eq. (\ref{basis}) in next subsection).

\subsection{Min-entropy}
Considering min-entropies in order to quantify the uncertainty of a measurement is meaningful for several reasons: on one hand $H_{\infty}$ (see \eqref{minentro}) sets a lower bound on all other entropies within the Rényi-$\alpha$ family, i.e. $H_{\infty}\leq H_{\alpha}$ for all $\alpha\in\mathbb{R}_+$. On the other hand it has a direct operational interpretation in the following sense: consider again an  experimenter who now samples a (discrete) random variable $Y_\varepsilon^I$ and assume that she tries to guess the outcome of a particular sample. In this case the quantity $\exp\left(-H_{\infty}(Y_\varepsilon^I)\right)$ gives the highest guessing probability she can attain when doing so.
In the same spirit, the min-entropy $h_\infty(w)$, i.e. the continuous counterpart/analogue of $H_{\infty}$, of a probability density $w$ is given as
\begin{align}
h_\infty (w) = - \log \left(  \operatorname{ess}\sup_x | w(x) | \right).
\label{minentro2}
\end{align}
Here the essential supremum $ \operatorname{ess}\sup_x | w(x) |$ \eqref{minentro2} is needed to correctly deal with sets of measure zero.
In particular, when regarding a partitioning of an interval into bins in the limit of vanishing bin size, the essential supremum arises as the natural limit of a supremum over finite bins.  Operationally, this quantity can therefore be understood as follows: consider we choose a partitioning into bins $\Omega_i$ of size $|\Omega_i|$ in a measurement scenario in which the outcomes are distributed according to a probability density $w(x)$. If the experimenter was to guess the bin in which the next measurement outcome will occur, the success probability $\mathbb{P}_{i}=\int_{\Omega_i} w(x)$ is upper bounded by $\mathbb{P}_{i}\leq c^* |\Omega_i|$ with $c^*$ some constant that may depend on the binning sizes, but is independent of the bin $i$. The min-entropy characterises the smallest constant $c^*$ for all possible binning sizes which is exactly given by $\operatorname{ess}\sup_x | w(x) |$.

Using the framework developed in this paper we are able to directly compute the minimal length in terms of min-entropy
\begin{align}
  \label{mindef}
  \Gamma_{\min}^{\infty}&:=\inf_{\psi\in \mathcal{P}} \left[-\log \left( \operatorname{ess}\sup_x |\phi(x)|^2 \right) \right]\\
  &= -\log \left( \sup_{\psi\in \mathcal{P}}\operatorname{ess}\sup_x |\phi(x)|^2 \right)\ ,
\end{align}
as follows: consider the ``particle in a box'' basis for the interval $[-k_{\max},k_{\max}]$, i.e.
\begin{align}
  \psi_n(k)&=\frac{1}{\sqrt{k_{\max}}} \sin\left[\frac{\pi n}{2 k_{\max}}(k-k_{\max})\right] \ ,
  \label{basis}
\end{align}
with its Fourier transform
\begin{align}
  \phi_n(x)&=\left(\frac{\pi n^2 k_{\max}}{2}\right)^{1/2} \frac{\sin(k_{\max} x - \frac{\pi n}{2 })}{k_{\max}^2 x^2 - \frac{\pi^2 n^2}{4}} \mathrm{e}^{-\frac{\pi n}{2} i } \ .
\end{align}
We can then decompose any $\psi\in \mathcal{P}$ by $\psi(k)= \sum_n \alpha_n \psi_n(k)$ with a square summable sequence $\alpha_n$.
In the same way,
its Fourier transform $\phi=\mathcal{F}[\psi]$ reads $\phi(x)= \sum_n \alpha_n \phi_n(x)$.
We therefore have
\begin{align}
  \label{min2}
  \Gamma_{\min}^{\infty}&= -\log \left( \sup_{\alpha: ||\alpha||_{l^2}=1}\sup_x \big|\sum_n \alpha_n \phi_n(x)\big|^2 \right) \ ,
\end{align}
where we used the fact that the $\phi_n(x)$ are smooth to replace the essential supremum with the ordinary supremum. The term $|\sum_n \alpha_n \phi_n(x)|^2$ can be rewritten as a scalar product $|\langle \alpha_n|\phi_n(x)\rangle_{l^2}|^2$ which, due to the Cauchy-Schwarz inequality, is maximised by $\langle \phi_n(x)|\phi_n(x)\rangle^2$.
Also note that by choosing the phase of $\psi$ appropriately we can without loss of generality set the maximum of the $\phi_n(x)$ to be at $x=0$.
Some algebra shows that $\sum_n |\phi_n(0)|^2 =k_{\max}/\pi$, proving that the minimal length in terms of min-entropy is given by (see Fig. \ref{fig:hmin})
\begin{align}
\label{min4}
\Gamma_{\min}^{\infty}&= -\log(k_{\max}/\pi) \ .
\end{align}
In particular,  whenever the variance-based minimal length is normalised, i.e. $l^2_{\min}=1$, the minimal length in terms of min-entropy is given bys
\begin{align}
\Gamma_{\min}^{\infty}=\log 2 \equiv 1\,[\mathrm{bit}] \ ,
\end{align}
i.e. the {\it minimal length is exactly one bit}.
\begin{figure}
	\centering
\def\svgwidth{0.4\textwidth}
\input{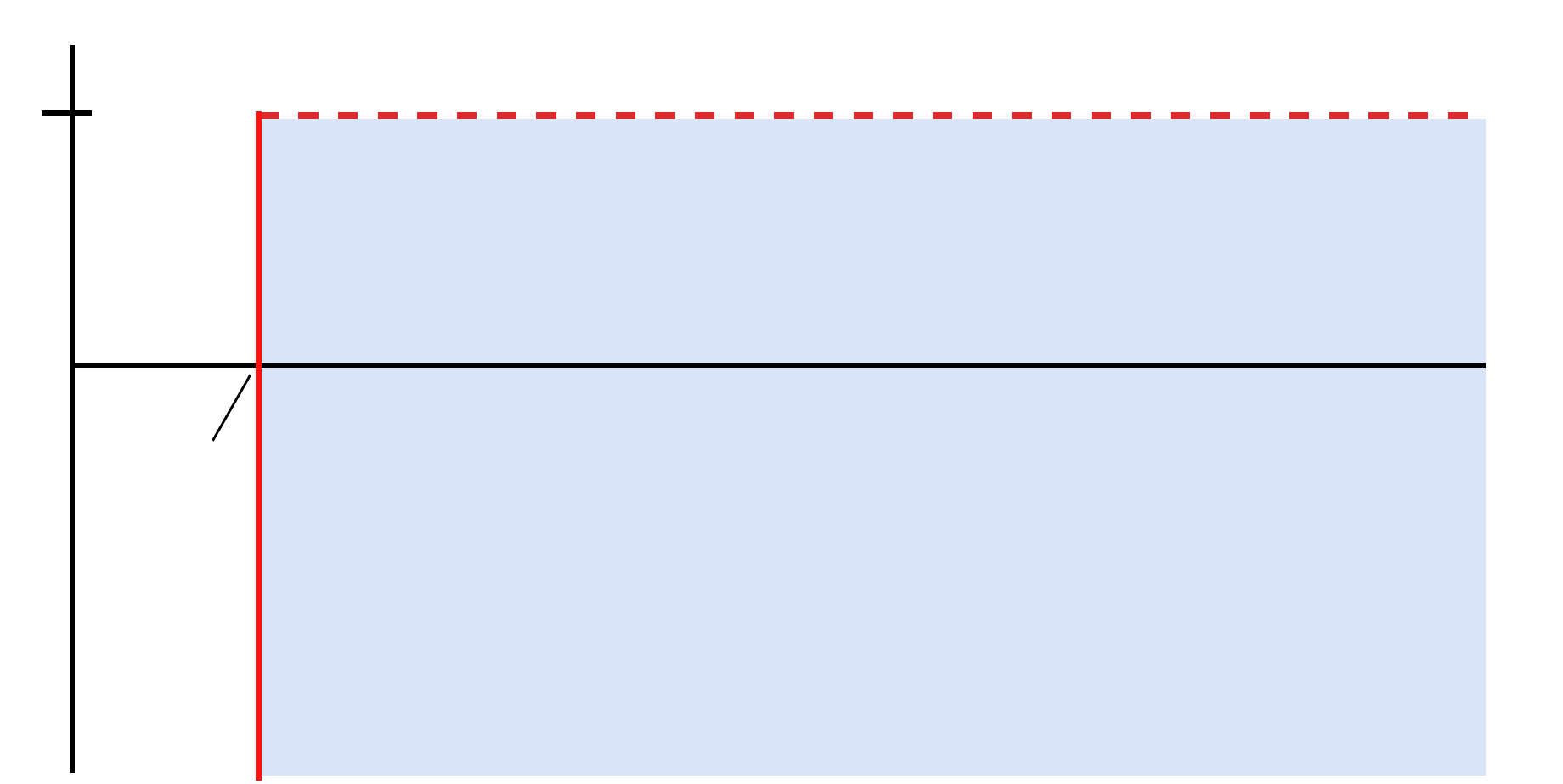}
\caption{Entropic uncertainty region for min-entropies $h_\infty$. The minimal length $\Gamma_{\min}^{\infty}$ is analytically computed in \eqref{min4}. As for Shannon entropies, the min-entropy of momentum attains an upper bound (dashed line), $\max h_{\infty}^p=\log(2 k_{\max})$, for the uniform distribution on the interval $I$. Numerical evaluation indicates that any entropy pair satisfying these two constraints can be attained. Hence the uncertainty region $\mathcal{U}_{h_{\infty}}$ and the corresponding trade-off curve are given by the shaded region and the solid line, respectively.}
\label{fig:hmin}
\end{figure}

%
%
%
%
\par\noindent

%
%
%
%
\par\noindent

\section{Concluding Remarks and Future Directions}
Minimal length is to be understood as the minimal possible uncertainty about position measurements due to a modification of the Heisenberg algebra. Nevertheless, two questions remained to be settled: first, under what assumptions does a modification actually {\it imply} minimal length? This line of thought is in contrast to previous attempts where only the existence of operators was shown which both show minimal length and satisfy the algebra. Despite its importance this question has not been answered so far to the best of our knowledge. One of our main results, Theorem \ref{thm:mainthm}, clarifies this issue and proves under physically well-motivated assumptions the uniqueness of such operators. As such it justifies results obtained in previous literature.

Second, one should operationally motivate the choice of uncertainty measure used to define minimal length, since such measures are far from unique. The choice to use variances is in this context not at all an obvious one. Instead entropic measures are known to have an operational interpretation while not suffering from a number of severe deficiencies that variances show if interpreted as a measure of uncertainty. We therefore introduce and show implications of an entropic formulation of minimal length.

Our main results can be summarized as follows: the states which correspond to the maximal possible localization in the $x$ space in terms of variance $\Delta \x$ and of the Shannon entropy $h_{\x}$ or min-entropy $h_{\infty}^x$, are different, showing that the physical notion of minimal length itself depends on the particular choice adapted to its operational meaning. The min-entropy lower bounds all other R\'{e}nyi-$\alpha$ entropies; likewise the minimal length in terms of min-entropy is a lower bound to all other entropic minimal lengths and turns out to be exactly one bit. On the other side, entropic bounds also show another novel feature, namely the presence of an upper bound on the entropy in $p$ space $h_\p$ or $h_\infty^p$ for "physical states". In other words deformations of the standard Heisenberg algebra leading to a minimal length lead to a lower limit on the information we can get on a given state in terms of its momentum distribution. This is not the case if we use variances to quantify this information, since the value of $\Delta p$ can be arbitrarily large.

We have established a framework that allows to compute optimal and state-independent uncertainty relations for modified Heisenberg algebras (see e.g. Theorem \ref{thm:convex} and Corollary \ref{cor:method}). Optimal and state-independent uncertainty relations directly yield minimal lengths, but contain much more information as they describe the uncertainties for all quantum states in the modified algebra.

One of the natural generalizations would be to extend our setting to higher dimensions. We hope that the study of the entropic uncertainties in three and four dimensions may shed some light on the plausible connection that might exist between previous limits such as bound on information storage (holographic bound)\cite{bekenstein1981universal, hooft1993dimensional, susskind1995world}, bounds on information scrambling/chaos\cite{Maldacena:2015waa}, bounds on quantum evolution \cite{anandan1990geometry,aharonov1961time,margolus1998maximum} and quantum computation/complexity\cite{lloyd2000ultimate,brown2015complexity,lloyd2012quantum}.

Furthermore our new entropic bounds on $h_\x$ and $h_\p$ are directly applicable to the entropic steering inequalities formulated in \cite{Schneeloch} and thus lead to new limitations on the amount of entanglement that can be shared between two distant parties governed by a modified Heisenberg algebra.

It is worth noting that one of our main results, Theorem \ref{thm:mainthm}, directly links the study of quantum physics in a modified algebra to the study of classical information processing of bandlimited analog signals. Quantum states are then replaced by the complex current in a wire, the considered observables change from position and momentum to time and frequency (again linked by Fourier transformation). A momentum cut-off due to a modified algebra can therefore be understood as a frequency (or ``band'') limitation of the complex current. We already exploited this analogy to some extent by considering operationally more relevant uncertainty measures as was first done in the well-studied field \cite{Slepian1,Slepian2,Slepian3} of classical information processing. By considering the Nyquist sampling theorem steps into this direction have been  taken by \cite{kempf-sampling}.
However, we strongly believe that one can obtain many more fundamental insights in the field of modified algebras by just transferring results and concepts from classical information theory.


\par\noindent

%
%
%
%
\no {\bf Acknowledgments}
%
%
\par\noindent
The authors are grateful to Saurya Das, Reinhard F. Werner, Tobias Osborne, and Inken Siemon for discussions related to this work. They especially thank Lars Dammeier for perpetually emphasising the importance of entropic uncertainty relations. They also thank Achim Kempf for reading a preliminary version of the manuscript and valuable comments. Special thanks also to Michael Hall for helpful comments on the definition of length in terms of entropies. This work was supported by the European grants DQSIM, QFTCMPS and SIQS and by the cluster of excellence EXC 201 Quantum Engineering and Space-Time Research. GM acknowledges the support of INFN, I.S. TASP. RS acknowledges the support from the BMBF project Q.com-q.    

%
%
%
\begin{appendix}
	\section*{ Appendix}
\section{Proof of Theorem \ref{thm:mainthm}}
  Starting point of our examination is the Hilbert space $\mc H:=\mc L^2(\mb R, \mr dx)$, witht the usual Borel-Lebesgue measure, together with the standard position operator $\fat x$ with dense domain $\mc D(\fat x)=\{\psi\in\mc H\,|\, \|\fat x\psi\|_2<\infty\}$.
  Recall that $\fat x$ is self-adjoint on $\mc D(\fat x)$ and has continuous spectrum $\sigma(\fat x)=\mb R$.
  The standard momentum operator $\fat k$ on $\mc H$ with its domain $\mc D(\fat k)$ of weakly differentiable functions in $\mc H$ is related to $\fat x$ by the Fourier-transform on $\mc H$ and it is self-adjoint. They fulfill the standard Heisenberg commutation relation $[\fat x,\fat k]\psi=i\psi$ for $\psi$ in a dense domain of analyic vectors for both operators.

  Our aim is to examine the situation where there is another linear and self-adjoint operator $\fat p$ on $\mc H$, possibly unbounded and with domain $\mc D(\fat p)$ such that there exists a sufficiently well-behaved function $f:\mb R\to\mb R$ such that
  \begin{align}\label{eq:defheiser}
    [\fat x,\fat p]\psi=if(\fat p)\psi
  \end{align}
  for some sensible choice of $\psi\in\mc H$.

  More precisely we even expect the operator $\fat p$ having domain $\mc D(\fat p)$ just on a sub-Hilbert space $\mc P\subset\mc H$ and such that it is self-adjoint only in $\mc P$ and the commutation relation~\eqref{eq:defheiser} only holds true for suitable vectors $\psi\in\mc P$.
  Even more, we want to show that under suitable assumptions on the function $f$ we can find a function $p:I\subseteq\mb R\to\mb R$ such that $\fat p=p(\fat k)$ and $I$ is a possibly unbounded interval. The Hilbert space $\mc P$ can then be identified with $\mc L^2(I)$.

  It is important here that we consider an embedding of $\mc P$ into the larger space $\mc H=\mc L^2(\mb R)$ since this allows us later to interpret the operator $\fat p$ as a modified version of the momentum operator.

  \subsection{An example}
    We start with an example from which we extract the essential structure we use later in our assertion.
    For this consider, as above, the Hilbert space $\mc H=\mc L^2(\mb R,\mr dx)$, and standard position and momentum operators $\fat x$ and $\fat k$ satisfying the standard canonical commutation relations $\fat x\fat k\psi-\fat k\fat x\psi=i\psi$ on a common dense set $\mc D$ of analytic vectors $\psi$.
    Denote by $E_{\fat k}$ the spectral measure of $\fat k$. Note that both operators, $\fat k$ and $\fat x$, have simple spectrum $\sigma(\fat k)=\mb R=\sigma(\fat x)$. That is, they are unitarily equivalent to a position operator on some $\mc L^2(\mb R,\mu)$, where $\mu$ is a suitable measure \cite{AkhiezerGlazman,Schmuedgen:2012}.
    We consider the representation, where $\fat k$ (the momentum operator) acts on $\mc H$ as the position operator, i.e. $\forall\psi\in\mc D(\fat k): (\fat k\psi)(k)=k\psi(k)$, and where $\mc D(\fat k)$ is the domain of $\fat k$.

    Now consider the real-valued function $p$ on $\mb R$ given by $p(k)=\mr{tan}\circ\chi_{I}(k),k\in\mb R$, where $\chi_{I}$ is the characteristic function on the interval $I=[-\frac\pi2,\frac\pi2]$.
    From this we get a self-adjoint operator
    \begin{align*}
      \fat p  &:= p(\fat k)=\int_{\mb R}p(\lambda)E_{\fat k}(\mr d\lambda).
    \end{align*}
    with domain $\tilde{\mc D}(\fat p)=\{\psi\in\mc H\,|\,\int_{\mb R}|p(\lambda)|^2\mr d\langle\psi,E_{\fat k}(\lambda)\psi\rangle<\infty\}$
    and $\mr d\langle\psi,E_{\fat k}(\lambda)\psi\rangle$ is the unique measure (on the Borel-$\sigma$-algebra $\mc B(\mb R)$) given by $\langle\psi,E_{\fat k}(\cdot)\psi\rangle$.

    Let $\mc P:=\mc L^2(I,\mr dx)\subset\mc H$ and denote by $P$ the projection $\mc H\to\mc P$. Then it is easy to check that $\fat p=p(P\fat kP)$, and, when regarded as an operator, $\fat p$ is self-adjoint on the domain $\mc D(\fat p)=\tilde{\mc D}(\fat p)\cap\mc P$.
    Note that the operator $P\fat kP$ is bounded, since the function $k\mapsto k$ is bounded on $I$.

    Denote the operator $P\fat kP$, when regarded as operator on the Hilbert space $\mc P$, by $\fat k_P$.
    Then $\fat k_P$ has simple spectrum $\sigma(\fat k_P)=\mb R$, and the self-adjoint, unbounded operator $\fat p$, as operator on $\mc P$ has spectrum $\sigma(\mb R)$.
    Furthermore, as such, $\fat p=p(\fat k_P)$, and since the function $\mr{tan}$ is bijective on $I$, $\sigma(\fat p)$ is simple. Note that, when regarded as operator on $\mc H$ the spectrum is not simple anymore, since zero is then an infinitely degenerate eigenvalue of $\fat p$.
    In any case the spectral projections of $\fat p$ commute with the spectral projections of $\fat k$ and $\fat k_P$, i.e. the operators $\fat k_P$ and $\fat p$ strongly commute.

    Since the tangent is analytic the dense subspace of analytic vectors $\mc D$ for $\fat p$ in $\mc P$ is contained in the space of analytic vectors for $\fat x$. For example, the smooth, compactly supported functions in $I$ that exponentially decay at the boundary of $I$ are analytic for $\fat p$, and also for $\fat x$.
    Moreover, for $\psi\in\mc D$ we have that $\fat x\fat p\psi-\fat p\fat x\psi=i(1+\fat p^2)\psi$.

    Now let $U:\mb R\times\mc H\to\mc H$ be the unitary 1-parameter group of translations generated by $\fat x$.
    Then for $\Omega\in\mc B(\mc R)$ and $t\in\mb R$ we have that $U_tE_{\fat k}(\Omega)U_{-t}=E_{\fat k}(\Omega+t)$, where $\Omega+t=\{x+t\,|\, x\in\Omega\}$.
    If $g:\mb R\to\mb R$ is an analytic function we get $(U_tE_{\fat k}(\Omega)U_{-t}\psi)(k)=\chi_{\Omega}(g(k-t))\psi(k)=\chi_{\Omega_g}(k)\psi(k)$, where we set $\Omega_g:=\{k\in\mb R\,|\,g(k-t)\in\Omega\}=g^{-1}(\Omega)+t$.
    Hence $U_tE_{\fat p}(\Omega)U_{-t}=E_{\fat k}(\Omega_p)=E_{\fat p}(\Omega')+Q$ with $\Omega':=p(\Omega_{p})\cap[-\frac\pi2,\frac\pi2]$ and $Q=E_{\fat k}(\Omega_{p}\setminus[-\frac\pi2,\frac\pi2])$.
    This can also be used to see that for all $\epsilon\leq\frac\pi2$ there exists an $\epsilon'>0$ such that $E_{\fat k}([-\epsilon,\epsilon])=E_{\fat p}([-\epsilon',\epsilon'])$ by simply choosing $\epsilon'=\arctan\epsilon$.

  \subsection{The general case}
    \begin{assumptions}
      Given $\mc H$, $\fat x$, $\fat k$ as before. Let $f:\mb R\to\mb R$ be a function with the following properties:
      \begin{itemize}
        \item
          $f$ is smooth.
        \item
          $f(0)=1$.
        \item
          $f$ is symmetric, i.e. $\forall p\in\mb R:f(-p)=f(p)$.
        \item
          $f$ is convex on $\mb R^+$.
      \end{itemize}
      Furthermore there exists a closed subspace $\mc P\subseteq\mc H$ with projection $P:\mc H\to\mc P$ such that $(\fat p,\mc D(\fat p))$ satisfies:
      \begin{itemize}
        \item
          $\mc D(\fat p)\subset\mc P$
        \item
          $(\fat p,\mc D(\fat p))$ is a self-adoint, linear operator on $\mc P$
        \item
          The spectrum $\sigma(\fat p)$ is continuous, coincides with $\mb R$, and is simple. I.e. there exists a vector $\psi\in\mc P$ such that for any other vector $\phi\in\mc P$ there exists a function $f\in\mb L^2(\sigma(\fat p),\mu)$ such that $\phi=\int_{\mb R}f(t)\mr dE_{\fat p}(t)\psi$ and $\mu$ is the measure given by $\mu(\Omega)=(E_{\fat p}(\Omega)\psi,\psi)$.
        \item
          There exists a dense subspace $\mc D\subset\mc P$ such that:
          $\forall\psi\in\mc D:\fat x\fat p\psi-\fat p\fat x\psi=if(\fat p)\psi$.
        \item
          For all $\psi\in\mc D$ and for all $n\in\mb N$ it holds that $\fat x^n\psi\in\mc D$ and $\fat p^n\psi\in\mc D$. In other words, $\mc D$ is a dense set of analytic vectors in $\mc P$ for both, $\fat x$ and $\fat p$.
        \end{itemize}
      We say that the objects $(f,\fat p,\mc D(\fat p),\mc P)$ are \emph{admissible} if they satisfy all of the above assumptions.
    \end{assumptions}
    The existence of dense subsets of analytic vectors for $\fat x$, respectively $\fat p$ follows from the simplicity of their spectra \cite[Theorem 69.3]{AkhiezerGlazman}. We want, however, that there exists a \emph{common} set of analytic vectors for both operators, which is contained just in $\mc P$.

    By \cite[Theorem 69.2]{AkhiezerGlazman} (or \cite[Proposition 5.18]{Schmuedgen:2012}) the assumption that $\fat p$ has simple spectrum implies that $\mc P\cong\mc L^2(\mb R,\mu)$, that is, vectors $\psi\in\mc P$ correspond to (equivalence classes of) functions $\psi\in\mc L^2(\mb R,\mu)$. Furthermore, by abuse of notation, $\forall\psi\in\mc D(\fat p): (\fat p\psi)(p)=p\psi(p)$.
    \begin{definition}[\cite{ReSi:1970}]
      Let $A$ and $B$ be unbounded self-adjoint operators on some Hilbert space $\mc H$ and $E_A$ and $E_B$ their projection-valued measures over $\mb R$. We say that $A$ and $B$ \emph{strongly commute}, if for all measurable sets $\Omega,\Omega'\subset\mb R$ the according spectral projections commute, i.e. $E_A(\Omega)E_B(\Omega')=E_B(\Omega')E_A(\Omega)$.
    \end{definition}
    \begin{theorem}\label{thm:thm1}
      Given the standard position and momentum operators $\fat x$ and $\fat k$ on $\mc H=\mc L^2(\mb R)$ and given $(f,\fat p,\mc D(\fat p),\mc P)$ admissible. Denote by $\mc B(\mb R)$ the Borel-$\sigma$-algebra on $\mb R$ and by $E_{\fat k},E_{\fat p}$ the spectral measures for $\fat k$ and $\fat p$, respectively.

      Consider the following conditions:
      \begin{itemize}
        \item
          There exists $\epsilon_0$ such that $\forall\,0<\epsilon\leq\epsilon_0\,\exists\epsilon'>0$ and $\exists\delta>0$ it holds:
          \begin{align}\label{eq:comm}
            \|E_{\fat k}([-\epsilon,\epsilon])-E_{\fat p}([-\epsilon',\epsilon'])\|<\delta
          \end{align}
          and $\delta\in O(\epsilon^3)$.
        \item
            Given the strongly-continuous 1-paramter group $U_t$ of translations generated by $\fat x$ we require that for any measurable set $\Omega\in\mc B(\mc R)$ there exists a measureable set $\Omega'\in\mc B(\mb R)$ and a projection $Q\leq\mb I-P$, where $P$ is the projection onto $\mc P$, such that $U_{t}E_{\fat p}(\Omega)U_{-t}=E_{\fat p}(\Omega')+Q$.
      \end{itemize}
      Then it follows that $\fat k$, when restricted to $\mc P$, has a self-adjoint extension $\fat k_P$ on $\mc P$ and that $\fat p$ and $\fat k_P$ strongly commute as operators on $\mc P$.
    \end{theorem}
    By \cite[Corollary 5.28]{Schmuedgen:2012} there then exists a dense linear subspace in $\mc P$ invariant under both operators $\fat p$ and the restricted $\fat k$, such that they commute on vectors from this domain.
    \begin{proof}
      The idea is to show that the spectral projections commute for any pair of finite Borel sets $\Omega,\Omega'\in\mc B(\mc H)$. These sets can be partitioned into finitely many intervals of diameter $\epsilon$ and $\epsilon'$.
      The spectral projections of $\fat k$ for these intervals are simply translates of $[-\epsilon',\epsilon]$ by the unitary group $U_t$ generated by the position operator $\fat x$.
      The idea now is to show that we can approximate the projections $U_{-t}E_{\fat p}([-\epsilon,\epsilon])U_{t}$ by spectral projections $E_p(\tilde\Omega)$ for some $\tilde\Omega\in\mc B(\mb R)$ which depends on $\epsilon,\epsilon'$ and $t$.
      Note that since $\mc B(\mb R)$ is generated by open and bounded intervals, and since spectral measures are countably additive the assertion holds for all measurable sets if its holds for bounded open intervals.

      Let $\Omega,\Omega'\in\mc B(\mb R)$ be bounded and open intervals. Choose any $\epsilon,\epsilon'>0$ and let $\Omega_\epsilon^\alpha$ and $\Omega_{\epsilon}^\beta$ be a cover of $\Omega$ and $\Omega'$ by disjoint open intervals $\Omega_\epsilon^\alpha:=(-\epsilon+\alpha,\epsilon+\alpha)$ and $\alpha,\beta$ are the corresponding indices for the shifts.
      By countable additivity it follows
      \begin{align*}
        [E_{\fat k}(\Omega),E_{\fat p}(\Omega')]  &=  \sum_{\alpha,\beta}[E_{\fat k}(\Omega_{\epsilon}^\alpha),E_{\fat p}(\Omega_{\epsilon}^\beta)].
      \end{align*}
      The unitary group $U_t$ generated by $\fat x$ acts as translations on $L^2(\mb R)$ when we view $\fat k$ as multiplication operator.
      Hence we have for the spectral projections $E_{\fat k}(\Omega)$ for some measurable set $\Omega\in\mc B(\mb R)$ that $U_tE_{\fat k}(\Omega)U_{-t}=E_{\fat k}(\Omega+t)$. I.e. for each index $\alpha$ we get $E_{\fat k}(\Omega_\epsilon^\alpha)=U_\alpha E_{\fat k}(\Omega_\epsilon)U_{-\alpha}$ where $\Omega_\epsilon:=\Omega_\epsilon^0$.

      Now choose $\epsilon<\epsilon_0$ and let $\epsilon'$ and $\delta$ such that $\|E_{\fat k}(\Omega_\epsilon)-E_{\fat p}(\Omega_{\epsilon'})\|<\delta$. Choose index sets $I,J$ such that $\Omega\subset\bigcup_{\alpha\in I}\Omega_\epsilon^\alpha=:\Omega_\epsilon^I$ and $\Omega'\subset\bigcup_{\beta\in J}\Omega_{\epsilon'}^\beta=:\Omega_{\epsilon'}^J$.
      It is enough to show that the commutators $[E_{\fat k}(\Omega_\epsilon^I),E_{\fat p}(\Omega_{\epsilon'}^J)]$ are small since this implies that this is also true for the smaller projections $E_{\fat k}(\Omega)$ and $E_{\fat p}(\Omega')$.

      By assumption we get for all such $\alpha$ and with the notation $\Omega_{\epsilon'}:=\Omega_{\epsilon'}^0$
      \begin{align*}
        \|E_{\fat k}(\Omega_\epsilon^\alpha)-U_\alpha E_{\fat p}(\Omega_{\epsilon'})U_{-\alpha}\|<\delta.
      \end{align*}
      Furthermore, we have that for each $\alpha\in I$ there exists a measureable set $\Omega_{\epsilon'}^\alpha$ and a projections $Q_\alpha\leq\mb I-P$ such that
      \begin{align*}
        U_\alpha E_{\fat p}(\Omega_{\epsilon'})U_{-\alpha}  &=  E_{\fat p}(\Omega_{\epsilon'}^\alpha)+Q_\alpha.
      \end{align*}
      Hence, we obtain the following estimate:
      \begin{align*}
        \|[E_{\fat k}(\Omega),E_{\fat p}(\Omega')]\|  &\leq  \sum_{\alpha,\beta}\|[E_{\fat k}(\Omega_{\epsilon}^\alpha),E_{\fat p}(\Omega_{\epsilon}^\beta)]\|\\
          &\leq \sum_{\alpha,\beta}\left (\|[U_\alpha E_{\fat p}(\Omega_{\epsilon'})U_{-\alpha},E_{\fat p}(\Omega_\epsilon^\beta)\|+2 \delta\right )\\
          &= \sum_{\alpha,\beta}\left [E_{\fat p}(\Omega_{\epsilon'}^\alpha)+Q_\alpha,E_{\fat p}(\Omega_\epsilon^\beta]\|+2 \delta\right )\\
          &=  \sum_{\alpha,\beta}2\delta =  2|J||I|\delta\leq\frac{4|\Omega||\Omega'|}{\epsilon^2}\delta.
      \end{align*}
      Since this estimate only depends on the partitioning of $\Omega$ and $\Omega'$ into small intervals and since we can choose this partitioning arbitrarily small, it follows that
      \begin{align*}
        \forall\Omega,\Omega'\in\mc B(\mb R)  &:  \|[E_{\fat k}(\Omega),E_{\fat p}(\Omega')]\|=0.
      \end{align*}
      In particular this implies that for all measurable sets $\Omega\in\mc B(\mb R)$ the projections $E_{\fat k}(\Omega)$ commute with the projection $P:\mc H\to\mc P$.
      By \cite[Theorem 75.1]{AkhiezerGlazman} this implies that $P$ is a function of $\fat k$, i.e. there exists a measurable function $\chi_P$ on $\mb R$ such that $P=\int_{\mb R}\chi_P(\lambda)E_{\fat k}(\mr d\lambda)$.
      Furthermore, $\chi_P$ is positive, $\chi_P^2=\chi_P$ and $\|\chi_P\|_\infty=1$, hence it is an indicator function of some measurable set $I\in\mc B(\mb R)$, and the projection $P$ coincides with $E_{\fat k}(I)$.
      The restriction of $\fat k$ to $\mc P$ is therefore given by
      \begin{align*}
        \fat k_P  &=  \int_{\mb R}\chi_P(\lambda)\lambda E_{\fat k}(\mr d\lambda).
      \end{align*}
    \end{proof}

    \begin{proposition}
      Given the standard position and momentum operators $\fat x$ and $\fat k$ on $\mc H=\mc L^2(\mb R)$ and given $(f,\fat p,\mc D(\fat p),\mc P)$ admissible, satisfying the assumptions in Theorem~\ref{thm:thm1}.
      Then there exists an interval $I\subseteq\mb R$ such that $\mc P=\mc L^2(I)$, and a function $p:I\to\mb R$ such that $\fat p=p(\fat k)$ on $\mc P$. The interval is determined by the function $f$ in the following sense:
      \begin{itemize}
        \item
          If the function $g(p)=\int_0^p\frac1{f(s)}\mr ds$ is bounded then
          $I=[-k_{max},k_{max}]$ and $k_{max}$ is given by
          \begin{align*}
            k_{max} &=  \int_0^\infty\frac1{f(p)}\mr dp.
          \end{align*}
        \item
          If $g$ is unbounded then $I=\mb R$.
      \end{itemize}
    \end{proposition}
    \begin{proof}
      By the previous theorem the operators $\fat k_P$ and $\fat p$ strongly commute on $\mc P$. Let $I\subset\mb R$ the measureable set with $E_{\fat k}(I)=P$ and $P:\mc H\to \mc P$ the projection onto $\mc P$.
      By assumption the spectrum of $\fat p$ is simple and
      therefore the spectral projections $E_{\fat k}(\Omega)$ with $\Omega\subset I$ are bounded functions of $\fat p$ \cite[Theorem VII.5]{ReSi:1970}. Since for any such $\Omega$ $E_{\fat k}(\Omega)$ is a projection there exists a measureable set $\Theta\subset \mb R$ such that $E_{\fat k}(\Omega)=E_{\fat p}(\Theta)$.
      Therefore there exists an almost everywhere finite, measureable, real-valued function $g$ on $\mb R$ such that
      \begin{align*}
        P\fat kP  &=  \int_{\mb R}g(\lambda)E_{\fat p}(\mr d\lambda),
      \end{align*}
      i.e. $P\fat k P=g(\fat p)$.
      Now let $\psi\in\mc D$. Then, by assumption
      \begin{align*}
        [\fat x,\fat p]\psi &=  if(\fat p)\psi.
      \end{align*}
      Then
      \begin{align*}
        [\fat x,g(\fat p)]\psi  &=  ig'(\fat p)f(\fat p)\psi.
      \end{align*}
      But, since $[\fat x,\fat k]\psi=i\psi$ we must have that $g'=\frac1f$, therefore
      \begin{align*}
        g(p)  &=  \int_{0}^p \frac1{f(s)}\mr ds.
      \end{align*}
      Since the spectrum $\sigma(g(\fat p))$ of $g(\fat p)$ is the essential range of the function $g$ we see that the spectrum of $P\fat kP$ is the interval $I=[-k_{max},k_{max}]$, if $g$ is bounded, and $\mb R$ if $g$ is unbounded.
      Hence the subspace $\mc P$ is isomorphic to $\mc L^2(I)$.

      Conversely, since $\fat k$ has simple spectrum, there exists a function $h$ such that $\fat p = h(\fat k)$, and this function is necessarily unbounded.
    \end{proof}

\section{Proof of Theorem \ref{thm:convex}}
\label{robbound}
We consider observables $\x$ and $\p$ that obey a modified commutation relation as described is Section \ref{heisenalg}. Since we are interested in uncertainty relations that are optimal for all states, we need to consider mixed states as well. In order to focus on pure states one would need to first show that for all mixed states there exists a pure state that is ``more optimal'', a notion that we will make precise in the following.\\
We denote by $\Omega$ the set of mixed states $\rho$, given by the density operators from $\mathcal{B}(\mathcal{P})$. When considering variances, the corresponding uncertainty region is given by
\begin{align}
\mathcal{U}=\left\{\left(\Delta_\rho \x , \Delta_\rho \p\right)\in \mathbb{R}^2_+\bigr|\rho\in\Omega\right\}.
\end{align}
We can give a precise definition of the trade-off curve by introducing a partial ordering relation ``$<$'' (``$\leq$''), which is given by saying that $v<w$ ($v\leq w$), for $v,w\in\mathbb{R}^2$, if every component of $v$ is smaller (small or equal) than the corresponding component of $w$. The trade-off curve $\Gamma(\mc U)$, i.e. the desired optimal and state-independent uncertainty relation, is then given by all tuples from $\mc U$ that are minimal in $\mc U$ with respect to the above ordering ``$<$'', i.e.
\begin{align}
  \label{bounddef}
  \Gamma(\mc U)=\left\{ v\in\mathcal{U}\bigr|\nexists w\in\mathcal{U}: w<v \right\}.
\end{align}
Unfortunately, no general efficient method for computing this trade-off curve is known. However, in Theorem \ref{boundthm}, we circumvent this circumstance by first providing a lower bound on $\Gamma(\mc U)$ and then showing that this bound can be attained. For the first step we need the notion of a lower convex hull $\mathcal{U}_{lc}$:
This is obtained by first filling up $\mathcal{U}$ with all points that are more uncertain than, at least, some point from $\mathcal{U}$, and then taking the convex hull of this set, i.e.
\begin{align}
  \label{lchull}
  \mathcal{U}_{lc}=\operatorname{Conv}\left(\left\{v\in\mathbb{R}^2\bigr|\exists w\in\mathcal{U}: w\leq v\right\}\right).
\end{align}
Note that this does not add any additional extremal points other than those already contained in the convex hull of $\mc U$ (which are therefore already contained in $\mc U$).

\begin{figure}
	\centering
	\def\svgwidth{0.45\textwidth}
	\input{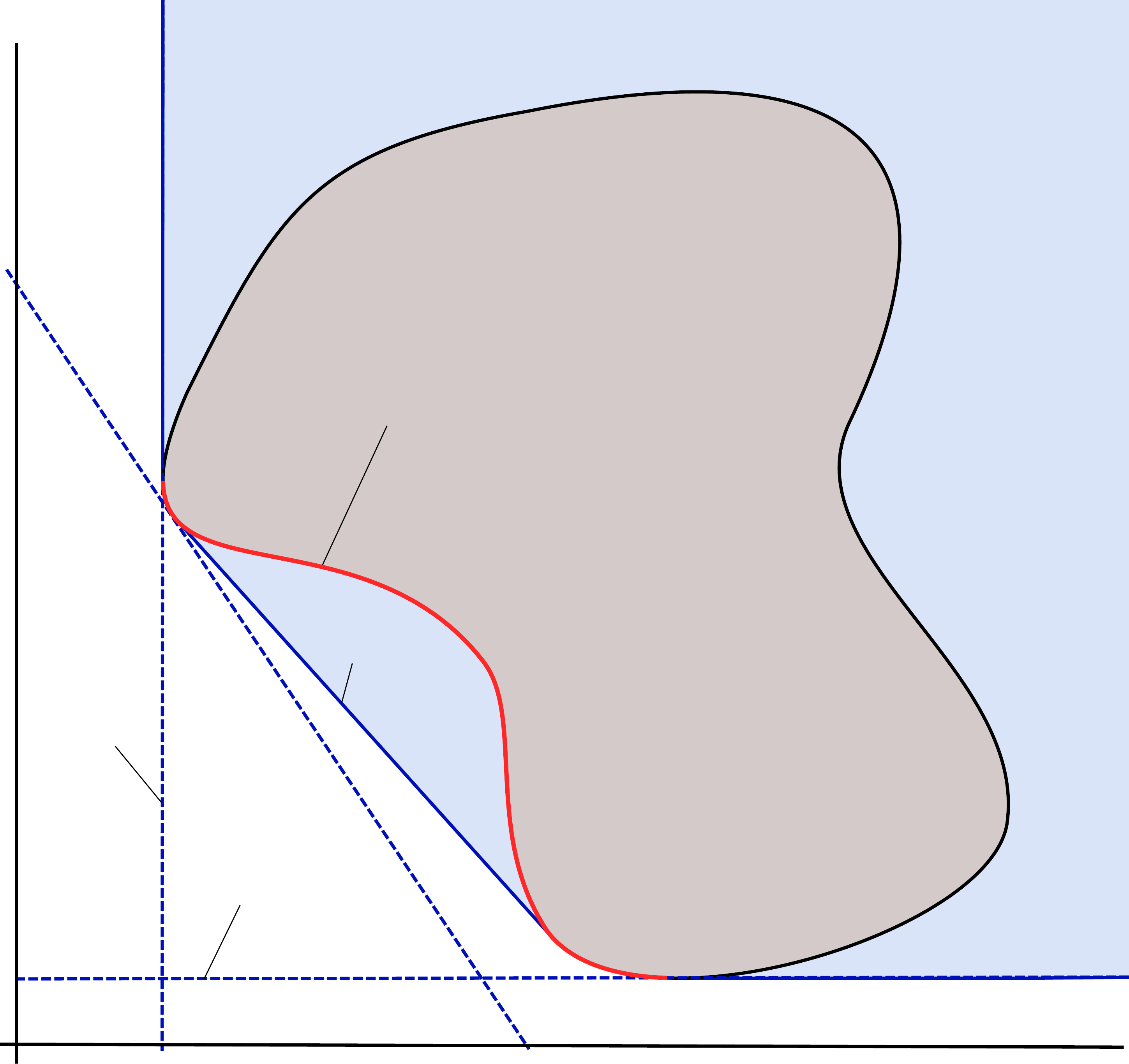}
	\caption{The lower convex hull $\mathcal{U}_{lc}$ (blue region) of some set $\mathcal{U}$ (grey region) is obtained by considering hyperplanes (dashed blue lines) for $\lambda\in(0,1)$. As in \eqref{linbound} the trade-off curve $\Gamma(\mathcal{U}_{lc})$ (solid blue line) of the convex hull can be parametrised by values $u(\lambda)$. The trade-off curve $\Gamma(\mathcal{U}_{lc})$ is never ``worse'' than the trade-off curve $\Gamma(\mathcal{U})$ (solid red line); both curves coincide if $\mathcal{U}$ is convex itself.}
	\label{fig:convex}
\end{figure}

\begin{theorem}
  \label{boundthm}
  Let $\x$, $\p$, $\Omega$, $\mc U$ and $\mc U_{lc}$ be given as described above and let $\mc U_{00}$ denote the set of attainable tuples of second moments, i.e.
  \begin{align}
  \label{f00}
    \mc U_{00}=\left\lbrace(\tr(\rho \x^2),\tr(\rho \p^2))|\rho\in\Omega\right\rbrace \ .
  \end{align}
  The notation indicates that $\mc U_{00}$ corresponds to the set obtained by restricting $\mc U$ to states with expectations $\langle \x\rangle=\langle \p\rangle=0$.
  Then
  \begin{align}\label{equalityU}
    \Gamma(\mathcal{U}_{lc})=\Gamma(\mc U)=\Gamma(\mc U_{00}) \ .
  \end{align}
\end{theorem}
\begin{proof}
For any two sets $\mc V_1$ and $\mc V_2$ we introduce a partial ordering relation ``$\preceq$'' on the corresponding curves $\Gamma(\mc V_1)$ and $\Gamma(\mc V_2)$, respectively, by saying that $\Gamma(\mc V_1) \preceq \Gamma(\mc V_2)$
if and only if for all points $v_2\in\Gamma(\mc V_2)$ there is a point $v_1\in\Gamma(\mc V_1)$ such that $v_1\leq v_2$.

By construction we know that $\mathcal{U}_{lc}$ is convex and $\Gamma(\mathcal{U}_{lc})\preceq\Gamma(\mc U)$.
Note that $\mc U_{00}$ is also convex since it is obtained as the range of an affine map of the convex set $\Omega$. Additionally, points in $\Gamma(\mc U_{00})$ are always in $\mc U$ and hence
\begin{align}
  \Gamma(\mathcal{U}_{lc})\preceq\Gamma(\mc U)\preceq\Gamma(\mc U_{00}) \ .
\end{align}
In the following we will prove that in fact $\Gamma(\mathcal{U}_{lc})=\Gamma(\mathcal{U}_{00})$, which immediately implies the desired statement \eqref{equalityU}.

Since $\mathcal{U}_{lc}$ is convex, it is fully characterised by the intersection of its supporting halfspaces (see for example \cite{convexopti}), i.e. there is a function $u(\lambda)$ such that
\begin{align}\label{dual}
  \mathcal{U}_{lc}=\left\{(v_1,v_2)\in\mathbb{R}_+^2|\forall \lambda \in \mathbb{R}:\lambda v_1 +(1-\lambda)v_2\geq u(\lambda) \right\}
\end{align}
The boundary of $\mathcal{U}_{lc}$ is the set of points which have an intersection with a supporting hyperplane, i.e. $\lambda v_1 +(1-\lambda)v_2 = u(\lambda)$.
Moreover, $\Gamma(\mathcal{U}_{lc})$ consists of points on the boundary, hence that attain equality for $\lambda\in(0,1)$.
Conversely, $u(\lambda)$ can be obtained by minimizing $\lambda v_1+(1-\lambda)v_2$ over all points in $\mathcal{U}_{lc}$, for a fixed $\lambda$.
However, for $\lambda$ in $[0,1]$, this minimization of a linear functional will attain its minimum also on extremal points of $\mathcal{U}_{lc}$, which are contained in the boundary of $\mathcal{U}$. We therefore can write
\begin{align}
  \label{min1}
  u(\lambda)=\inf_{\rho\in\Omega} \lambda \Delta_\rho\x+(1-\lambda) \Delta_\rho\p\ .
\end{align}
At this stage we can reformulate the variance of an observable $A$ as
\begin{align}
  \Delta_\rho A = \min_{a\in\mathbb{R}} \langle (A-a)^2\rangle_\rho\ ,
\end{align}
such that the r.h.s of \eqref{min1} turns into
\begin{align}
  \label{min2app}
  \min_{\alpha,\eta\in\mathbb{R}} \inf_{\rho\in\Omega}\lambda \langle(\x-\eta)^2\rangle_\rho +(1-\lambda) \langle(\p-\alpha)^2\rangle_\rho\ .
\end{align}
The expectation of $\x$ can be shifted by multiplying with $\exp( i \eta \k)$. As $\p$ commutes with $\k$, (see section \ref{heisenalg}) this procedure will not affect the variance $\Delta \p$, such that we can always set $\eta=0$ in the following.
If we now represent $\p$ as a function of the coordinate $k$ and $\x$ as $i\partial_k$, the minimization over $\rho$ in \eqref{min2app} corresponds to finding the ground state energy, $E_\alpha$, of the Schrödinger operator
\begin{align}
H_\alpha:=-\lambda \partial^2_k +V_\alpha(k)
\end{align}
with potential $V_\alpha(k)=(1-\lambda) (p(k)-\alpha)^2$, i.e.
\begin{align}
u(\lambda)=\min_\alpha E_\alpha.
\end{align}

As $V_\alpha(k)$ is positive and thus bounded from below for every $\alpha$, there is a unique function $\check{V}_\alpha(k)$ which gives the best convex approximation to $V_\alpha(k)$ from below, i.e the super graph of  $\check{V}_\alpha(k)$ is the convex hull of the super graph of $V_\alpha(k)$. If needed, $\check{V}_\alpha(k)$ can be obtained by Legendre transforming $V_\alpha(k)$ twice. Now, for all states $\rho$, we have $\langle V_\alpha(k) \rangle_\rho\geq\langle \check{V}_\alpha(k) \rangle_\rho$ and hence we can lower bound $E_\alpha $ by $\check{E}_\alpha$, which is the ground state energy of the Schrödinger operator $-\lambda \partial^2_k +\check{V}_\alpha(k)$.

Note that $\check{V}_\alpha(k)$ is a convex function in $\alpha$, because $V_\alpha(k)$ is convex in $\alpha$. We can therefore employ Corollary 13.6 from \cite{simon} to show that $\check{E}_\alpha$ is a convex function of $\alpha$, too.

Moreover, $\check{V}_\alpha$ inherits the symmetry $\check{V}_\alpha(-k)=\check{V}_{-\alpha}(k)$ from $V_\alpha(k)$, which can be implemented by a unitary automorphism on $\Omega$. This shows the symmetry $\check{E}_\alpha=\check{E}_{-\alpha}$, which directly implies that $\check{E}_\alpha$ becomes minimal for $\alpha=0$. But here we have that $\check{V}_0(k)=V_0(k)$, because $V_0(k)$ is already a convex function, which yields
\begin{align}
\label{min3}
\check{E}_0=E_0=u(\lambda).
\end{align}

We thus know that, for $\lambda\in(0,1)$, the extremal points of $\mc U_{lc}$ (which are in $\Gamma(\mc U_{lc})$) have zero expectation in $\x$ and $\p$, i.e. they lie within $\mathcal{U}_{00}$. Since $\Gamma(\mc U_{lc})\preceq \Gamma(\mc U_{00})$ we even know that on these points the boundaries $\Gamma(\mc U_{lc})$ and $\Gamma(\mc U_{00})$ coincide. But as $\mc U_{lc}$ and $\mc U_{00}$ are both convex, this implies that $\Gamma(\mathcal{U}_{lc})=\Gamma(\mathcal{U}_{00})$.
\end{proof}

\section{A state-independent but not so optimal bound}
\label{nonoptbound}

Assume a modification $f(p)$ of the Heisenberg algebra with Taylor expansion $f(p)=1+\sum_{n=1}^\infty a_np^{2n}$ and $a_n\geq0$ for all $n\in\mb N$.
This implies that $f$ is convex, monotonously increasing for $p\geq 0$, smooth and symmetric around the origin with $f(0)=1$, hence fulfilling our assumptions (i-iii) in the main text. If we set
\begin{align}
g(p):=f(-\sqrt{|p|} )\ ,
\end{align}
we can see that
\begin{align}
g(p^2)=f(-\sqrt{|p^2|} )=f(-|p|)=f(|p|)=f(p)\ . \label{squared}
\end{align}
Now, restricted to $p>0$, $g(p)$ arises as a concatenation of convex functions, thus $g$ is also convex in this parameter range.
Inserting \eqref{squared} into the Robertson Kennard relation \eqref{robxp}, and using  Jensen's inequality together with the convexity of $g$, we get
\begin{align}
\Delta\x\Delta\p
&\geq \frac14|\langle f(\p)\rangle|^2\geq\frac14|g(\langle \p^2 \rangle)|^2\nonumber \ .
\end{align}
Now we can substitute $\langle p^2 \rangle=\Delta\p +\langle \p\rangle^2$ and use the properties of $g$ as well as the convexity of the absolute square, to arrive after a simple calculation at
\begin{align*}
\Delta\x\Delta\p  &\geq \frac14|g(\Delta\p+\langle \p\rangle^2)|^2\\
  &\geq \frac14|g(\Delta \p)|^2+\frac14|g(\langle \p\rangle^2)|^2-1\ .
\end{align*}
Here the state-dependent term $\frac14|g(\langle \p\rangle^2)|^2$ is greater than or equal to one, such that we can conclude the state-independent bound
\begin{align}
\Delta\x\Delta\p\geq\frac14g(\Delta \p)^2 \ .
\end{align}
However this bound is not optimal for general modifications. Examples for this can be seen in Fig.~\ref{ucrbounds}.

\end{appendix}
\bibliography{GEUP}

\begin{thebibliography}{10}

\bibitem{Amati:1988tn}
D.~Amati, M.~Ciafaloni, and G.~Veneziano.
\newblock Can space-time be probed below the string size?
\newblock {\em Physics Letters B}, 216:41, 1989.

\bibitem{garay1995quantum}
L.~J. Garay.
\newblock Quantum gravity and minimum length.
\newblock {\em International Journal of Modern Physics A}, 10(02):145--165,
  1995.
\newblock {arXiv}: gr-qc/9403008.

\bibitem{kempf1995hilbert}
A.~Kempf, G.~Mangano, and R.~B. Mann.
\newblock {Hilbert} space representation of the minimal length uncertainty
  relation.
\newblock {\em Physical Review D}, 52(2):1108, 1995.
\newblock {arXiv}: hep-th/9412167.

\bibitem{hossenfelder2013minimal}
S.~Hossenfelder.
\newblock Minimal length scale scenarios for quantum gravity.
\newblock {\em Living Reviews in Relativity}, 16(2):90, 2013.
\newblock {arXiv}: 1203.6191.

\bibitem{Konishi:1989wk}
K.~Konishi, G.~Paffuti, and P.~Provero.
\newblock Minimum physical length and the generalized uncertainty principle in
  string theory.
\newblock {\em Physics Letters B}, 234:276, 1990.

\bibitem{Kempf:1993bq}
A.~Kempf.
\newblock {Uncertainty relation in quantum mechanics with quantum group
  symmetry}.
\newblock {\em Journal of Mathematical Physics}, 35:4483--4496, 1994.
\newblock {arXiv}: hep-th/9311147.

\bibitem{connes}
A.~Connes.
\newblock {\em Noncommutative Geometry}.
\newblock Academic Press, 1995.

\bibitem{Moyal:1949sk}
J.~E. Moyal.
\newblock {Quantum mechanics as a statistical theory}.
\newblock {\em Mathematical Proceedings of the Cambridge Philosophical
  Society}, 45:99--124, 1949.

\bibitem{Groenewold:1946kp}
H.~J. Groenewold.
\newblock On the principles of elementary quantum mechanics.
\newblock {\em Physica}, 12:405--460, 1946.

\bibitem{Doplicher:1994tu}
S.~Doplicher, K.~Fredenhagen, and J.~E. Roberts.
\newblock {The Quantum structure of space-time at the {Planck} scale and
  quantum fields}.
\newblock {\em Communications in Mathematical Physics}, 172:187--220, 1995.
\newblock {arXiv}: hep-th/0303037.

\bibitem{Mangano:2015pha}
G.~Mangano, F.~Lizzi, and A.~Porzio.
\newblock Inconstant {Planck}'s constant.
\newblock {\em International Journal of Modern Physics A}, 30(34):1550209,
  2015.
\newblock {arXiv}: 1509.02107.

\bibitem{Pikovski:2011zk}
I.~Pikovski, M.~R. Vanner, M.~Aspelmeyer, M.~S. Kim, and C.~Brukner.
\newblock {Probing {Planck}-scale physics with quantum optics}.
\newblock {\em Nature Physics}, 8:393--397, 2012.
\newblock {arXiv}: 1111.1979.

\bibitem{Wilde}
M.~Wilde.
\newblock {\em Quantum Information Theory}.
\newblock Cambridge University Press, 2013.
\newblock {arXiv}: 1106.1445.

\bibitem{Goold:2015}
J.~Goold, M.~Huber, A.~Riera, L.~del Rio, and P.~Skrzypczyk.
\newblock The role of quantum information in thermodynamics --- a topical
  review.
\newblock {\em Journal of Physics A: Mathematical and Theoretical},
  49(14):143001, 2016.
\newblock {arXiv}: 1505.07835v2.

\bibitem{carroll2016entropy}
S.~M. Carroll and G.~N. Remmen.
\newblock What is the entropy in entropic gravity?
\newblock {\em Physical Review D}, 93:124052, 2016.
\newblock {arXiv}: 1601.07558.

\bibitem{Furrer}
F.~Furrer, T.~Franz, M.~Berta, A.~Leverrier, V.~B. Scholz, M.~Tomamichel, and
  R.~F. Werner.
\newblock Continuous variable quantum key distribution: Finite-key analysis of
  composable security against coherent attacks.
\newblock {\em Physical Review Letters}, 109:100502, 2012.
\newblock {arXiv}: 1112.2179v3.

\bibitem{dammeier2015uncertainty}
L.~Dammeier, R.~Schwonnek, and R.~F. Werner.
\newblock Uncertainty relations for angular momentum.
\newblock {\em New Journal of Physics}, 17(9):093046, 2015.
\newblock {arXiv}: 1505.00049.

\bibitem{Busch:2013vba}
P.~Busch, P.~Lahti, and R.~F. Werner.
\newblock Measurement uncertainty relations.
\newblock {\em Journal of Mathematical Physics}, 55:042111, 2014.
\newblock {arXiv}: 1312.4392.

\bibitem{abdelkhalek2015optimality}
K.~Abdelkhalek, R.~Schwonnek, H.~Maassen, F.~Furrer, J.~Duhme, P.~Raynal, B.~G.
  Englert, and R.~F. Werner.
\newblock Optimality of entropic uncertainty relations.
\newblock {\em International Journal of Quantum Information}, 13(06):1550045,
  2015.
\newblock {arXiv}: 1509.00398.

\bibitem{coles2015entropic}
P.~J. Coles, M.~Berta, M.~Tomamichel, and S.~Wehner.
\newblock Entropic uncertainty relations and their applications.
\newblock 2015.
\newblock {arXiv}: 1511.04857.

\bibitem{WehnerWinter}
S.~Wehner and A.~Winter.
\newblock Entropic uncertainty relations -- a survey.
\newblock {\em New Journal of Physics}, 12(2):025009, 2010.
\newblock {arXiv}: 0907.3704.

\bibitem{Pye:2015tta}
J.~Pye, W.~Donnelly, and A.~Kempf.
\newblock {Locality and entanglement in bandlimited quantum field theory}.
\newblock {\em Physical Review D}, 92(10):105022, 2015.
\newblock {arXiv}: 1508.05953.

\bibitem{Pareto}
Altannar Chinchuluun, Panos~M. Pardalos, Athanasios Migdalas, and Leonidas
  Pitsoulis.
\newblock {\em Pareto Optimality, Game Theory And Equilibria}.
\newblock Springer New York, 2008.

\bibitem{Brout}
R.~Brout, Cl. Gabriel, M.~Lubo, and Ph. Spindel.
\newblock Minimal length uncertainty principle and the trans-planckian problem
  of black hole physics.
\newblock {\em Physical Review D}, 59(4), 1999.

\bibitem{Werner1990}
R.~F. Werner.
\newblock Dilations of symmetric operators shifted by a unitary group.
\newblock {\em Journal of Functional Analysis}, 92(1):166--176, 1990.

\bibitem{Robertson}
H.~P. Robertson.
\newblock The uncertainty principle.
\newblock {\em Physical Review}, 34:163--164, 1929.

\bibitem{Kennard}
E.~H. Kennard.
\newblock Zur {Q}uantenmechanik einfacher {B}ewegungstypen.
\newblock {\em Zeitschrift f{\"u}r Physik}, 44:326--352, 1927.

\bibitem{Bojowald:2011jd}
M.~Bojowald and A.~Kempf.
\newblock {Generalized uncertainty principles and localization of a particle in
  discrete space}.
\newblock {\em Physical Review D}, 86:085017, 2012.
\newblock {arXiv}: 1112.0994.

\bibitem{deutsch1983uncertainty}
D.~Deutsch.
\newblock Uncertainty in quantum measurements.
\newblock {\em Physical Review Letters}, 50(9):631, 1983.

\bibitem{spindel}
S.~Detournay, C.~Gabriel, and P.~Spindel.
\newblock About maximally localized states in quantum mechanics.
\newblock {\em Physical Review D}, 66:125004, 2002.
\newblock {arXiv}: hep-th/0210128.

\bibitem{dorsch}
G.~C. Dorsch and J.~A. Noguera.
\newblock Maximally localized states in quantum mechanics with a modified
  commutator relation to all orders.
\newblock {\em International Journal of Modern Physics A}, 27(21):1250113,
  2012.
\newblock {arXiv}: 1106.2737.

\bibitem{convexopti}
S.~Boyd and L.~Vandenberghe.
\newblock {\em Convex Optimization}.
\newblock Cambridge University Press, 2004.

\bibitem{srw}
R.~Schwonnek, D.~Reeb, and R.~F. Werner.
\newblock Measurement uncertainty for finite quantum observables.
\newblock {\em Mathematics}, 4(2):21, 2016.
\newblock {arXiv}: 1604.00382.

\bibitem{ReSi:1970}
M.~Reed and B.~Simon.
\newblock {\em Functional Analysis}, volume~1 of {\em Methods of Modern
  Mathematical Physics}.
\newblock 1970.

\bibitem{teller}
G.~P{\"o}schl and E.~Teller.
\newblock Bemerkungen zur {Q}uantenmechanik des anharmonischen {O}szillators.
\newblock {\em Zeitschrift f{\"u}r Physik}, 83(3-4):143--151, 1933.

\bibitem{bialynicki2011entropic}
I.~Bia{\l}ynicki-Birula and {\L}.~Rudnicki.
\newblock Entropic uncertainty relations in quantum physics.
\newblock In {\em Statistical Complexity}, pages 1--34. 2011.

\bibitem{maassen1988generalized}
H.~Maassen and J.~Uffink.
\newblock Generalized entropic uncertainty relations.
\newblock {\em Physical Review Letters}, 60(12):1103, 1988.

\bibitem{shannon}
C.~E. Shannon.
\newblock A mathematical theory of communication.
\newblock {\em The Bell System Technical Journal}, 27(3):379--423, 1948.

\bibitem{renyi1961}
A.~R{\'e}nyi.
\newblock On measures of entropy and information.
\newblock In {\em Proceedings of the Fourth Berkeley Symposium on Mathematical
  Statistics and Probability, Volume 1: Contributions to the Theory of
  Statistics}, pages 547--561, Berkeley, Calif., 1961.

\bibitem{Jaynes63}
E.~T. Jaynes.
\newblock Information theory and statistical mechanics.
\newblock {\em Brandeis University Summer Institute Lectures in Theoretical
  Physics, Statistical Physics}, 3:181--218, 1963.

\bibitem{bialynicki1975uncertainty}
I.~Bia{\l}ynicki-Birula and J.~Mycielski.
\newblock Uncertainty relations for information entropy in wave mechanics.
\newblock {\em Communications in Mathematical Physics}, 44(2):129--132, 1975.

\bibitem{Bialynicki-B:1984}
Iwo Bia{\l}ynicki-Birula.
\newblock Entropic uncertainty relations.
\newblock {\em Physics Letters A}, 103(5):253--254, 1984.

\bibitem{beckner1975inequalities}
W.~Beckner.
\newblock Inequalities in {F}ourier analysis.
\newblock {\em Annals of Mathematics}, pages 159--182, 1975.

\bibitem{lieb}
E.~Lieb.
\newblock Gaussian kernels have only {G}aussian maximizers.
\newblock {\em Inventiones Mathematicae}, 102:179--208, 1990.

\bibitem{bekenstein1981universal}
J.~D. Bekenstein.
\newblock Universal upper bound on the entropy-to-energy ratio for bounded
  systems.
\newblock {\em Physical Review D}, 23(2):287, 1981.

\bibitem{hooft1993dimensional}
G.~'t~Hooft.
\newblock Dimensional reduction in quantum gravity.
\newblock 2009.
\newblock {arXiv}: gr-qc/9310026.

\bibitem{susskind1995world}
L.~Susskind.
\newblock The world as a hologram.
\newblock {\em Journal of Mathematical Physics}, 36(11):6377--6396, 1995.
\newblock {arXiv}: hep-th/9409089.

\bibitem{Maldacena:2015waa}
J.~Maldacena, S.~H. Shenker, and D.~Stanford.
\newblock A bound on chaos.
\newblock 2015.
\newblock {arXiv}: 1503.01409.

\bibitem{anandan1990geometry}
J.~Anandan and Y.~Aharonov.
\newblock Geometry of quantum evolution.
\newblock {\em Physical Review Letters}, 65(14):1697, 1990.

\bibitem{aharonov1961time}
Y.~Aharonov and D.~Bohm.
\newblock Time in the quantum theory and the uncertainty relation for time and
  energy.
\newblock {\em Physical Review}, 122(5):1649, 1961.

\bibitem{margolus1998maximum}
N.~Margolus and L.~B. Levitin.
\newblock The maximum speed of dynamical evolution.
\newblock {\em Physica D: Nonlinear Phenomena}, 120(1):188--195, 1998.
\newblock {arXiv}: quant-ph/9710043.

\bibitem{lloyd2000ultimate}
S.~Lloyd.
\newblock Ultimate physical limits to computation.
\newblock {\em Nature}, 406(6799):1047--1054, 2000.
\newblock {arXiv}: quant-ph/9908043.

\bibitem{brown2015complexity}
A.~R. Brown, D.~A. Roberts, L.~Susskind, B.~Swingle, and Y.~Zhao.
\newblock Complexity, action, and black holes.
\newblock {\em Physical Review D}, 93,:086006, 2015.
\newblock {arXiv}: 1512.04993.

\bibitem{lloyd2012quantum}
S.~Lloyd.
\newblock The quantum geometric limit.
\newblock 2012.
\newblock {arXiv}: 1206.6559.

\bibitem{Schneeloch}
J.~Schneeloch, C.~J. Broadbent, S.~P. Walborn, E.~G. Cavalcanti, and J.~C.
  Howell.
\newblock {Einstein-Podolsky-Rosen} steering inequalities from entropic
  uncertainty relations.
\newblock {\em Physical Review A}, 87:062103, 2013.
\newblock {arXiv}: 1303.7432.

\bibitem{Slepian1}
D.~Slepian and H.~O. Pollak.
\newblock Prolate spheroidal wave functions, {Fourier} analysis and uncertainty
  {I}.
\newblock {\em Bell System Technical Journal}, 40(1):43--63, 1961.

\bibitem{Slepian2}
H.~J. Landau and H.~O. Pollak.
\newblock Prolate spheroidal wave functions, {Fourier} analysis and uncertainty
  {II}.
\newblock {\em Bell System Technical Journal}, 40(1):65--84, 1961.

\bibitem{Slepian3}
H.~J. Landau and H.~O. Pollak.
\newblock Prolate spheroidal wave functions, {Fourier} analysis and uncertainty
  {III}.
\newblock {\em Bell System Technical Journal}, 41(4):1295--1336, 1962.

\bibitem{kempf-sampling}
A.~Kempf and R.~Martin.
\newblock Information theory, spectral geometry, and quantum gravity.
\newblock {\em Physical Review Letters}, 100:021304, 2008.
\newblock {arXiv}: 0708.0062.

\bibitem{AkhiezerGlazman}
N.~I. Akhiezer and I.~M. Glazman.
\newblock {\em Theory of linear operators in {H}ilbert space}.
\newblock Dover Publishing Inc., 1963.

\bibitem{Schmuedgen:2012}
K.~Schm{\"u}dgen.
\newblock {\em Unbounded self-adjoint operators on {H}ilbert space}.
\newblock Springer Science \& Business Media, 2012.

\bibitem{simon}
B.~Simon.
\newblock {\em Functional Integration and Quantum Physics}.
\newblock AMS, 1979.

\bibitem{Kempf:1997pb}
A.~Kempf.
\newblock On nonlocality, lattices and internal symmetries.
\newblock {\em Europhysics Letters}, 40:257--261, 1997.
\newblock {arXiv}: hep-th/9706213.

\bibitem{kempf1997minimal}
A.~Kempf and G.~Mangano.
\newblock Minimal length uncertainty relation and ultraviolet regularization.
\newblock {\em Physical Review D}, 55(12):7909, 1997.
\newblock {arXiv}: hep-th/9612084.

\bibitem{Hall}
M.~J. Hall.
\newblock Universal geometric approach to uncertainty, entropy, and
  information.
\newblock {\em Physical Review A}, 59(4), 1999.

\end{thebibliography}
\bibliographystyle{unsrt}
\end{document}